\documentclass[a4paper,11pt]{amsart} 

\usepackage{bbold}                                                           
\usepackage{diagbox}
\usepackage{multirow}
\usepackage{amssymb}
\usepackage{amsmath}
\usepackage{amsthm}
\usepackage{amsfonts}
\usepackage{mathtools}

\usepackage{tikz-cd}
\usepackage{color}
\usepackage{hyperref}
\usepackage{url}
\usepackage{multicol}
\usepackage{ulem}

\usepackage[utf8]{inputenc}
\usepackage[T1]{fontenc}

\newtheorem{thm}{Theorem}[section]
\newtheorem{rema}[thm]{Remark}
\newtheorem{cor}[thm]{Corollary}
\newtheorem{lem}[thm]{Lemma}
\newtheorem{prop}[thm]{Proposition}
\newtheorem{prope}[thm]{Properties}

\newtheorem{defn}[thm]{Definition}
\newtheorem{example}[thm]{Example}

\theoremstyle{definition}
\usepackage{enumerate}
\usepackage{frcursive}
\numberwithin{equation}{section}
\usepackage{mathrsfs}
\allowdisplaybreaks[1]
\usepackage{stmaryrd}

\newcommand{\R}{\mathbb R}

\renewcommand{\a}{\alpha}

\renewcommand{\b}{\beta}
\newcommand{\s}{\sigma}
\renewcommand{\o}{\omega}

\newcommand{\G}{\Gamma}

\newcommand{\dis}{\displaystyle}

\newcommand{\End}{\mathsf{End}}

\begin{document}

\title{Actions of Lie 2-algebras and comomentum  maps}
\author{Philippe Bonneau
\email{philippe.bonneau@univ-lorraine.fr}
,
V\'eronique Chloup
\email{veronique.chloup@univ-lorraine.fr}
,
Angela Gammella-Mathieu
\email{angela.gammella-mathieu@univ-lorraine.fr}
and 
\newline Tilmann Wurzbacher
\email{tilmann.wurzbacher@univ-lorraine.fr}
}
\address{Institut \'Elie Cartan de Lorraine, Universit\'e de Lorraine et CNRS, 57070 Metz, France}


\date{\today}

\begin{abstract}
In this paper we introduce the notion of a 2-action of a Lie 2-algebra on an arbitrary manifold $M$. Furthermore, in \cite{Rog}, given a $n$-plectic manifold $(M,\omega)$,
the authors consider a $L_{\infty}$-algebra
$\mathsf{L^{\infty}(M,\omega)}$, which is a higher analogue of the Poisson algebra of observables associated to a symplectic manifold. This $L_{\infty}$-algebra reduces to a Lie 2-algebra $\mathsf{L^2(M,\omega)}$ when $(M,\o)$ is 2-plectic.
Following ideas of N.L. Delgado \cite{Del}, we introduce the Lie 2-algebra $\mathsf{D^2(M,\o)}$, which generalises the Lie 2-algebra $\mathsf{L}^2(M,\o)$ and its extension $\mathsf{\widetilde{D^2(M,\omega)}}$ containing Hamiltonian pairs.
Given a 2-plectic manifold $(M,\omega)$ and a Lie 2-algebra $\mathfrak{g}_{-1}\oplus\mathfrak{g}_{0}$ acting on $M$ we define a comomentum  map as a lift of the action, i.e., as a Lie 2-algebra morphism from $\mathfrak{g}_{-1}\oplus\mathfrak{g}_{0}$ to the Lie 2-algebra $\mathsf{\widetilde{D^2(M,\omega)}}$.
In an appendix, we discuss very explicitly 
numerous examples, classified according to their algebraic properties.
\end{abstract}


\maketitle

\tableofcontents

\section{Introduction }
Lie group actions with a moment map play a crucial role in symplectic geometry. They are fundamental tools in Hamiltonian dynamics, as symmetries in the Noether sense yield conserved quantities or in more modern language give rise to the reduction of symplectic manifolds; they also define collective Hamiltonians allowing to factorise a Hamiltonian dynamical system. Furthermore, they are important in geometric or deformation quantization for the transition from classical systems with symmetries to quantum systems with symmetries, compare the famous ``quantization commutes with reduction"
paradigma.\\\\
The essential information of these symmetries is encoded in the comomentum map, a Lie algebra homomorphism from a typically finite dimensional Lie algebra ${\mathfrak{g}}_0$
to the Poisson Lie algebra of classical observables, i.e., the smooth functions on a symplectic manifold. Going from particle systems to field theories can often be achieved via replacing symplectic manifolds by multisymplectic manifolds. The latter class of manifolds coming with a non-degenerate, closed $k+1$-form, allows for ``Hamiltonian equations'' whose solutions are maps from $k$-dimensional sources with $k>1$ to this manifold, i.e., classical fields. In this context several analogs of the above mentioned observable algebras were proposed but all mathematically rigorous constructions have in common that these algebras are not Lie algebras but Lie $k$-algebras in the $k$-plectic case (cf. \cite{Rog, Zam} and for an overview \cite{RyWu1}). Of course, Lie groups as symmetry groups of differential forms are very natural and were early on considered in this context (compare \cite{CFRZ}, where the analogs of the comomentum maps were called homotopy moment maps). In a very special case \cite{Mam} and \cite{MaZa} gave a first definition for Lie 2-algebras but their ensueing actions on the manifold all factor through Lie algebra actions.
\\\\
Since Lie algebras are only a special case of Lie $k$-algebras, it is clear that the ``symmetries" of a $k$-plectic manifold $(M, \omega)$ should be rather formulated in terms of Lie $k$-algebras. Note that a Lie algebra action on a (symplectic) manifold is given by vector fields, i.e., derivations of the global observables algebra. By analogy, one is naturally led to think of vector fields plus bivector fields acting on functions and one-forms. In the 2-plectic case, the first beyond symplectic manifolds, this is the most na\"\i ve version of an algebra of observables (see the Appendix B below for some considerations along these lines).\\\\
In his substantial MSc thesis (available in preprint form in the arXiv, cf. \cite{Del}) N. Delgado defines the action of a Lie $k$-algebra on a $k$-plectic manifold. Unfortunately, Delgado gives only definitions and structural considerations
but no examples and no explicit formulae.\\\\
In this paper, we restrict ourselves to the 2-plectic case but give full algebraic and differential-geometric details of ``2-actions'', comomentum maps for Lie 2-algebra actions and a wealth of examples.
Our main contribution is the study of these actions and of comomentum maps (with values in Delgado's Lie 2-algebra of observables instead of the observables of J. Baez and C. Rogers as hitherto in the literature) for such actions on 2-plectic manifolds $(M,\omega)$.\\\\
Let us describe the content of this paper in more detail. 
In Section 2, we review the basics of Lie 2-algebras, their morphisms, and the composition of their morphisms, as well as the notions of skeletal and strict Lie 2-algebras. We recall the fact that each Lie 2-algebra is quasi-isomorphic to a skeletal Lie 2-algebra and we reprove the one-to-one correspondence between strict Lie 2-algebras and Lie algebra crossed modules in a direct way.
In the third section, we introduce a Chevalley-Eilenberg complex naturally associated to certain Lie 2-algebras (in particular to skeletal Lie 2-algebras). This enables us to associate to each Lie 2-algebra an explicit 3-cocycle  and to recover as a byproduct the classification of \cite{BaCra} and \cite{Wa1,Wa2}.
Section 4 is a reminder of the Cartan calculus for multivector fields that will be used in the later sections of the paper.\\\\
The Sections 5-7 contain the main novelties of this article.
Following N.L. Delgado, we equip in Section 5 the truncation 
$\mathfrak{X}^{\bullet \leq 2}(M)=
\mathfrak{X}^2(M)\oplus\mathfrak{X}^1(M)$ of the multivector fields on a manifold $M$ with a Lie 2-algebra structure and we introduce the notion of an action of a Lie 2-algebra (a 2-action), as a Lie 2-morphism  generalising the (infinitesimal) action of a Lie algebra $\mathfrak{g}_0$ on $M$. We observe here notably in Proposition 5.9 that quasi-isomorphisms are not compatible with 2-actions ! Strict 2-actions are also introduced and studied.\\
In Section 6, following again N.L. Delgado, we prove that the Lie 2-algebra $\mathsf{L^2(M,\omega)}$
of observables (\cite{Rog}) can be enlarged to a Lie 2-algebra $\mathsf{\widetilde{D^2(M,\omega)}}$ in case $M$ is a 2-plectic manifold. We underline here the different Lie 2-algebra morphisms, as the injection $I$ from $\mathsf{L^2(M,\omega)}$ to $\mathsf{\widetilde{D^2(M,\omega)}}$, the reverse surjection $\Phi$ and the ``symplectic 2-gradient" $\Psi : 
\mathsf{\widetilde{D^2(M,\omega)}}\to \mathfrak{X}^{\bullet \leq 2}(M)$. This latter morphism is rather implicit in \cite{Del} but central to the step from 1-plectic (aka symplectic) to 2-plectic geometry. Note also that we allow $\Psi_2$ to be non-zero but in \cite{Del} the corresponding morphism $\pi_1$ is always assumed to be strict.\\
In the last section, Section 7, we define a comomentum map associated to a 2-action of a Lie 2-algebra on a 2-plectic manifold $M$ as a lift of the action, satisfying a natural commutative diagram. 
This notion of a comomentum map (compare also Definition 6.11 in \cite{Del}) generalises the comomentum maps defined in \cite{CFRZ} for Lie algebras resp. in \cite{Mam, MaZa} for skeletal Lie 2-algebras (in the 2-plectic case).\\\\
Appendix A contains an explicit, constructive proof for the quasi-isomorphism between a given Lie 2-algebra and a skeletal one, illustrated by two fundamental examples. Appendix B gives a motivation for our constructions, via a linear representation of $\mathfrak{X}^{\bullet \leq 2}(M)$ on $C^\infty(M) \oplus \Omega^1(M)$, in analogy to the case of vector fields acting on functions. Appendix C contains numerous examples classified in function of the algebraic properties of the brackets of the Lie 2-algebras, the actions and the comomentum maps.

\section{Basic results on Lie 2-algebras}

In this section, the basic notions of the theory of Lie 2-algebras are recalled, and the correspondence between strict Lie 2-algebras and crossed modules of Lie algebras is reproved.

\subsection{\texorpdfstring{$L_\infty$}\ -algebras}

\begin{defn}
A \emph{$L_{\infty}$-algebra (or Lie ${\infty}$-algebra)} is a non-positively graded vector space $L$, endowed with a collection $\{l_k : L^{\otimes k}\mapsto L\}_{k\geq 1}$ of multilinear, graded antisymmetric maps such that $l_k$ is of degree $2-k$ and such that for all 
$m\geq 1$, 
$$\dis \sum_{i+j=m+1} (-1)^{i(j+1)} \dis \sum_{\s \in Sh(i,m-i)} (-1)^\s 
\epsilon(\s; x_1,..., x_m) l_j(l_i(x_{\s(1)},\ldots , x_{\s(i)}),x_{\s(i+1)},\ldots , x_{\s(m)})=0$$
where $\epsilon(\s; x_1,..., x_m)$ denotes the Koszul sign of $\s$ acting on the elements $x_1,..., x_m$ and $Sh(i,m-i)\subset S_m$ are the $(i,m-i)$-unshuffles.
\end{defn}

\begin{rema}
For $L_{\infty}$-algebras see \cite{Rog} and references therein.
\end{rema}

\begin{defn}
    A \emph{Lie $n$-algebra} is a $L_{\infty}$-algebra, concentrated in the degrees $1{-}n,\ldots,-1,0$.
\end{defn}

\subsection{Lie 2-algebras and their morphisms}
This subsection specialises from Lie $\infty$-algebras to Lie 2-algebras and gives an explicit description of these latter objects and their morphisms.\\

We will use the following explicit characterization of Lie 2-algebras in the sequel.
\begin{prop}
   A Lie 2-algebra is a graded vector space $L=L_{-1}\oplus L_0$, with a collection of three multilinear maps $(l_1,l_2,l_3)$, where
   $l_1:L_{-1}\mapsto L_0$, $l_2$ can be decomposed in its ``pure'' part
   $l_2^p:L_0\times L_0\mapsto L_0$, antisymmetric and its ``mixed'' part
   $l_2^m :L_{-1}\times L_0\mapsto L_{-1}$, also antisymmetric in the sense that 
   $l_2^m (x,a)=-l_2^m(a,x)$ for all $x$ in $L_0$ and $a$ in $L_{-1}$, and finally
   $l_3 : L_0\times L_0\times L_0\mapsto L_{-1}$, antisymmetric. (For convenience and when there is no ambiguity, we will simply use $l_2$, both for $l_2^p$ and $l_2^m$.)
   These maps satisfy the following relations, for all $a, b$ in $L_{-1}$ and for all $x,y,z,t$ in $L_0$
   \begin{itemize}
   \item [$(R_1)$]  $l_1(x)=0$
   \item [$(R_2)$]  $l_1(l_2(x,a))=l_2(x,l_1(a))$
   \item [$(R_3)$]  $l_2(l_1(a),b)=l_2(a,l_1(b))$
   \item [$(R_4)$] $l_1(l_3(x,y,z))=-l_2(l_2(x,y),z)-l_2(l_2(y,z),x)-l_2(l_2(z,x),y)$
   \item [$(R_5)$] $l_3(l_1(a),x,y)=-l_2(l_2(x,y),a)-l_2(l_2(y,a),x)-l_2(l_2(a,x),y)$
   \item [$(R_6)$] $l_3(l_2(x,y),z,t)-l_3(l_2(x,z),y,t)+ l_3(l_2(x,t),y,z)+l_3(l_2(y,z),x,t)-l_3(l_2(y,t),x,z)$\\
 $\quad \quad \quad +l_3(l_2(z,t),x,y)
  =l_2(l_3(x,y,z),t)-l_2(l_3(x,y,t),z)+l_2(l_3(x,z,t),y)-l_2(l_3(y,z,t),x).$
   \end{itemize}
\end{prop}

\begin{rema}
Every Lie algebra $(\mathfrak{g}, [\,,\,])$ can be seen as a Lie 2-algebra  $(L_{-1}\oplus L_0,l_1,l_2,l_3)$ upon setting $L_{-1}=\{0 \}$, $L_0=\mathfrak{g}$ and
$l_1=l_2^m=l_3=0$ and defining $l_2^p : L_0\times L_0 \mapsto L_0$ by
$$l_2^p(x,y)=[x,y]\, .$$
\end{rema}

\begin{rema} We also use notations such as $\mathfrak{g}=\mathfrak{g}_{-1}\oplus \mathfrak{g}_0$ for a Lie 2-algebra if useful and not confusing.
\end{rema}

\begin{defn}
    Let $L=(L_{-1}\oplus L_0,l_1,l_2,l_3)$ and 
    $L'=(L'_{-1}\oplus L_0',l_1',l_2',l_3')$ be two Lie 2-algebras. A \emph{Lie 2-algebra morphism or simply a Lie 2-morphism} $F$ from $L$ to $L'$ is a couple $F=(F_1,F_2)$  where each map $F_k$ is linear and of degree $1-k$. This couple can be decomposed into 
$F_{1,0} : L_0\mapsto L_0'$,
    $F_{1,-1} : L_{-1}\mapsto L_{-1}'$
    and 
    $F_{2} : L_0\times L_0\mapsto L_{-1}'$ (antisymmetric),
   satisfying the following relations, for all $x,y,z$ in $L_0$ and for all $a$ in $L_{-1}$ 
    \begin{itemize}
        \item [$(A_1)$] $l_1'\circ F_{1,-1}=F_{1,0}\circ l_1$
        \item [$(A_2)$] $l_1'(F_2(x,y))=F_{1,0}(l_2(x,y))-
        l_2'(F_{1,0}(x),F_{1,0}(y))$
        \item [$(A_3)$]$F_{1,-1}(l_2(a,x))=F_2(l_1(a),x)+
        l_2'(F_{1,-1}(a),F_{1,0}(x))$
       \item [$(A_4)$] $F_{1,-1}(l_3(x,y,z))+\big{(}F_2(l_2(x,y),z)\big{)}+ c.p.$\\
       $\quad \quad =l_3'(F_{1,0}(x),F_{1,0}(y),F_{1,0}(z))+\big{(}l_2'(F_{1,0}(x),F_2(y,z))+c.p.\big{)}$
       \par\noindent
       where $c.p.$ denotes cyclic permutations.
        
    \end{itemize}
\end{defn}

\begin{defn}
A Lie 2-morphism $F=(F_1,F_2)$ between two Lie 2-algebras is said to be \emph{strict} if $F_2=0.$
\end{defn}

Lie 2-algebras morphisms can be composed in the following way.
\begin{lem} \label{compositionmorphismes}
Let $L=(L_{-1}\oplus L_0,l_1,l_2,l_3)$,
$L'=(L'_{-1}\oplus L'_0,l'_1,l'_2,l'_3)$
and $(L''_{-1}\oplus L''_0,l''_1,l''_2,l''_3)$ be Lie 2-algebras.
Consider Lie 2-algebra morphisms $F=(F_1,F_2)$ from $L_{-1}\oplus L_0$
to $L'_{-1}\oplus L'_0$ resp. $F'=(F_1',F_2')$ from  $L'_{-1}\oplus L'_0$ to $L''_{-1}\oplus L''_0$. Then the composition 
$$F''=F'\circ F$$ is explicitly described as the couple $F''=(F''_{1},F''_2)$
where $F_{1,0}'':L_0\mapsto F_0''$, 
$F_{1,-1}'': L_{-1}\mapsto L_{-1}''$,
$F_2'':L_0\times L_0\mapsto L_{-1}''$
and
with the relations 
\begin{itemize}
\item [$(C_1)$]  $F_{1,0}''=F_{1,0}'\circ F_{1,0}$
   \item [$(C_2)$]  $F_{1,-1}''=F_{1,-1}'\circ F_{1,-1}$
   \item [$(C_3)$] $F_2''=F_2'\circ (F_{1,0}\times F_{1,0})+F_{1,-1}'\circ F_2.$
\end{itemize}
  \end{lem}
\begin{proof} Straightforward.
\end{proof}

The following result can also be easily deduced from the definitions.

\begin{prop}
  Let $(L_{-1}\oplus L_0,l_1,l_2,l_3)$  be a Lie 2-algebra. Then $L_0$ is a Lie algebra with the bracket induced by $l_2$ if and only if $l_1\circ l_3=0$.
\end{prop}

\subsection{Skeletal Lie 2-algebras}
This and the next subsection consider important special classes of Lie 2-algebras ($l_1=0$ resp. $l_3=0$) and recall the fundamental results on these classes : the quasi-isomorphy of any given Lie 2-algebra to a skeletal one resp. the correspondence between strict Lie 2-algebras and crossed modules of Lie algebras.
\begin{defn}
  A Lie 2-algebra $(L_{-1}\oplus L_0,l_1,l_2,l_3)$ is said to be \emph{skeletal (or minimal)} if $l_1=0.$  
\end{defn}

\begin{prop}
\label{quasi-iso}
Every Lie 2-algebra $(L_{-1}\oplus L_0,l_1,l_2,l_3)$ is quasi-isomorphic to a skeletal one.
\end{prop}
\begin{proof} 
A proof of this result can be found in \cite{BaCra}. 
In Appendix A, we will give a new constructive proof and apply our construction explicitly in fundamental examples. The proof can be skipped on first reading, but the construction will be used again in the following section. 
\end{proof}

\subsection{Strict Lie 2-algebras}\label{strict}

\begin{defn}
  A Lie 2-algebra $(L_{-1}\oplus L_0,l_1,l_2,l_3)$ is said to be \emph{strict} when $l_3=0$.
\end{defn}
\begin{rema}
A strict Lie 2-algebra is a differential graded Lie algebra, concentrated in the degrees 0 and -1.
\end{rema}

A useful equivalent description of strict Lie 2-algebras can be achieved in terms of Lie algebra crossed modules.

\begin{defn}
A \emph{Lie algebra crossed module} $(\mathfrak{g},\mathfrak{h},\tau,r)$ is given  by 
\begin{itemize}
\item two Lie 2-algebras $\mathfrak{g}$ and $\mathfrak{h}$, whose brackets
will be denoted by $[\,,\,]_{\mathfrak{g}}$ and
 $[\,,\,]_{\mathfrak{h}}$ 
\item 
two Lie algebra morphisms $\tau :\mathfrak{h}\mapsto \mathfrak{g}$
and $r : \mathfrak{g}\mapsto Der(\mathfrak{h})$
where $Der(\mathfrak{h})$ is the Lie algebra of derivations of 
$\mathfrak{h}$, with the bracket
$$[A,B]=A\circ B -B\circ A$$
for $A$ and $B$ in $Der(\mathfrak{h})$.
\end{itemize}

Furthermore, the maps $\tau$ and $r$ have to satisfy, for all $a,b$  in $\mathfrak{h}$, and for all $x$  in $\mathfrak{g}$,
\begin{itemize}
\item $\tau(r(x)(a))=[x,\tau(a)]_{\mathfrak{g}}$
\item $r(\tau(a))(b)=[a,b]_{\mathfrak{h}}$.
\end{itemize}
\end{defn}
As shown in \cite{BaCra}, Lie algebra crossed modules can be ``categorified'' as strict Lie 2-algebras. Let us sketch a proof of this result, which avoids categorical machinery. 
\begin{prop}
There is a one-to-one correspondence between strict Lie 2-algebras and Lie algebra crossed modules.
\end{prop}
\begin{proof} 
\hfill \newline
Let $(\mathfrak{g},\mathfrak{h},\tau,r)$ be a Lie algebra crossed module. We put $\mathfrak{g}_{-1}=\mathfrak{h}$ and $\mathfrak{g}_0=\mathfrak{g}$. We define the mappings $l_1$ and $l_2$ by  $l_1=\tau$, $l_2^p(x,y)=[x,y]_{\mathfrak{g}}$
and $l_2^m(a,x)=-l_2^m(x,a)=-r(x)(a)$ where $x$ and $y$ are in $\mathfrak{g}$, $a$ is in $\mathfrak{h}$, and we put $l_3=0$.
It is not difficult to prove that all the relations $(R_i)_{1\leq i\leq 6}$ are satisfied. Thus, with these mappings, $\mathfrak{g}_{-1}\oplus \mathfrak{g}_0$ is a Lie 2-algebra and this Lie 2-algebra is strict since $l_3=0$.\\ 

Conversely, let $\mathfrak{g}_{-1}\oplus \mathfrak{g}_0$ be a strict Lie 2-algebra with mappings $l_1$, $l_2$ which decomposes in $l_2^p$ and $l_2^m$ and with $l_3=0$.
We put $\mathfrak{h}=\mathfrak{g}_{-1}$,
$\mathfrak{g}=\mathfrak{g}_0$, $\tau=l_1$ and $r(x)(a)=l_2^m(x,a)$
where $x$ is in $\mathfrak{g}_0$ and $a$ is in $\mathfrak{g}_{-1}$. Moreover, we define the brackets $[\,,\,]_{\mathfrak{g}}$ and $[\,,\,]_{\mathfrak{h}}$ by
$$[x,y]_{\mathfrak{g}}=l_2^p(x,y)$$
and 
$$[a,b]_{\mathfrak{h}}=l_2^m(a, l_1(b))\, ,$$
where $x$ and $y$ are in $\mathfrak{g}_{0}$, $a$ and $b$ are in $\mathfrak{g}_{-1}$. One can easily check that $(\mathfrak{g}, [\,,\,]_{\mathfrak{g}})$ and $(\mathfrak{h}, [\,,\,]_{\mathfrak{h}})$ are Lie algebras. By simple computations, with the previous notations and with the relations of a Lie 2-algebra, we obtain also that, for all $x$ in $\mathfrak{g}$ and $a,b$ in $\mathfrak{h}$, 
$$r(x)([a,b]_{\mathfrak{h}})-[r(x)(a),b]_{\mathfrak{h}}-[a,r(x)(b)]_{\mathfrak{h}}=l_3(l_1(a), l_1(b),x)=0$$ since $l_3=0$. Thus $r(x)$ belongs to $Der(\mathfrak{h})$ for all $x$ in $\mathfrak{h}$. 
In the same way, we obtain that, for all $x$, $y$ in $\mathfrak{g}$ and $a$ in $\mathfrak{h}$ 
$$r([x,y]_{\mathfrak{g}})(a)-[r(x), r(y)]_{Der(\mathfrak{h})}(a)
=-l_3(l_1(a), x,y)=0$$
since $l_3=0$. Thus, $r :\mathfrak{g}\mapsto Der(\mathfrak{h})$ is a Lie algebra morphism. 
The remaining relations can now be checked implying that $(\mathfrak{g},\mathfrak{h},\tau,r)$ is indeed a Lie algebra crossed module.
\end{proof}

\begin{rema}
Let $\mathfrak{g}_{-1}\oplus \mathfrak{g}_0$ be a Lie 2-algebra fulfilling $l_1\circ l_3=0$ (so that $\mathfrak{g}_{0}$ is a Lie algebra) and 
$l_3(l_1(a), x,y)=0$ for all $a$ in $\mathfrak{g}_{-1}$ and $x$, $y$ in $\mathfrak{g}_{0}$. (These conditions are weaker than $l_3=0$.) Then, as in the second part of the previous proof, one can associate to such a Lie 2-algebra a Lie algebra crossed module. 
\end{rema}

\begin{defn}
 A \emph{crossed module morphism} between two Lie algebra crossed modules $(\mathfrak{g},\mathfrak{h},\tau,r)$ and 
 $(\mathfrak{g}',\mathfrak{h}',\tau',r')$  is a couple $(\Phi, \psi)$ of two Lie algebra morphisms, $\Phi : \mathfrak{h}\mapsto \mathfrak{h}'$ and $\psi : \mathfrak{g}\mapsto \mathfrak{g}'$ satisfying the following relations :
 \begin{itemize}
     \item $\tau'\circ \Phi=\psi \circ \tau$
     \item $\Phi(r(x)(a))=r'(\psi(x))(\Phi(a))$
 \end{itemize}
 for all $x$ in $\mathfrak{g}$ and $a$ in $\mathfrak{h}$.
 \end{defn}

It is now easy to prove the following result.
\begin{prop}
A crossed-module morphism between two Lie algebra crossed modules corresponds to a strict Lie 2-morphism between the two associated strict Lie 2-algebras.
\end{prop}

\section{Lie 2-algebras and 3-cocycles}
In this section we use the Chevalley-Eilenberg complex of a Lie algebra with values in a module to study Lie 2-algebras satisfying $l_1 \circ l_3=0$. It allows us to classify strict Lie 2-algebras (aka crossed modules of Lie algebras). This recovers classical results, compare, e.g., [Wa2] and [BaCra].

\medskip

Let $(L_{-1}\oplus L_0,l_1,l_2,l_3)$ be a Lie 2-algebra such that $l_1\circ l_3=0$
and $$l_3(l_1(a), x,y)=0$$
for all $a$ in $L_{-1}$, and for all $x$, $y$ in $L_0$. 
As remarked above, $L_0$ is then a Lie algebra and, furthermore, the mapping $p: L_0\mapsto End(L_{-1} )$, defined for $x$ in $L_0$ and $a$ in $L_{-1}$ by
$$p(x)(a)= l_2^m(a,x)$$
is a Lie algebra morphism between the Lie algebra $L_0$, endowed with the bracket induced by $l_2^p$ and 
the Lie algebra of endomorphisms $End(L_{-1})$, endowed with the commutator as its bracket, i.e. $p$ defines an action of $L_0$ on $L_{-1}$. To such a linear action, one can associate the Chevalley-Eilenberg cohomology of $L_0$ with values in $L_{-1}$, where the cochains are given by 
$$C^k(L_0,L_{-1})= Hom(\wedge^kL_0,L_{-1}).$$
For $c$ in $C^k(L_0,L_{-1})$, and for $x_1,...,x_{k+1}$ in $L_0$, we have the Chevalley-Eilenberg differential defined as
$$(d_{L_0}) c(x_1,...,x_{k+1})= \dis \sum_{i=1}^{k+1} (-1)^{i+1} p(x_i) (c(x_1,...,\widehat{x_i}, ..., x_{k+1}))
+ $$
$$\dis \sum_{i<j} (-1)^{i+j} c( l_2^p(x_i,x_j),x_1,...,\widehat{x_i},.., \widehat{x_j}, ..., x_{k+1}).$$
We denote the vector spaces of $k$-cocycles and $k$-coboundaries by ${Z}^k(L_0,L_{-1})$
and $B^k(L_0,L_{-1})$ respectively and define the $k$-th Lie algebra cohomology space by
$$H^k(L_0,L_{-1})= \dis \frac{Z^k(L_0,L_{-1})}{B^k(L_0,L_{-1})}.$$

\begin{rema}
 The assumptions $l_1\circ l_3=0$
and $$l_3(l_1(a), x,y)=0$$
for all $a$ in $L_{-1}$, and for all $x$, $y$ in $L_0$ are notably satisfied when $l_1=0$ (that is for a skeletal Lie 2-algebra). 
\end{rema}
\begin{prop}
With the assumptions of the preceding remark, $l_3$ is a 3-cocycle of the Chevalley-Eilenberg cohomology of $L_0$ with values in $L_{-1}$, that is 
$d_{L_0} l_3=0$.
\end{prop}

\begin{proof} Follows directly from relation $(R_6)$.
\end{proof}

\begin{prop}
Let $L=(L_{-1}\oplus L_0,l_1,l_2,l_3)$ be an arbitrary Lie 2-algebra and $\overline{L}=(\overline{L_{-1}}\oplus\overline{L_0}, \overline{l_1}=0, \overline{l_2}, \overline{l_3})$
the associated skeletal Lie 2-algebra. Then one can associate a 3-cocycle $\overline{l_3}$ of the Chevalley-Eilenberg cohomology of $\overline{L_0}$ with values in $\overline{L_{-1}}$, i.e., 
$d_{\overline{L_0}} \overline{l_3}=0$. 
\end{prop}

\begin{proof} In Appendix A, we show that $(L_{-1}\oplus L_0,l_1,l_2,l_3)$ is quasi-isomorphic to a skeletal Lie 2-algebra $(\overline{L_{-1}}\oplus\overline{L_0}, \overline{l_1}=0, \overline{l_2}, \overline{l_3})$. Now, using the preceding proposition for $\overline{L}$, $\overline{l_1}=0$ implies that $\overline{l_3}$ is a 3-cocycle for the Chevalley-Eilenberg cohomology of
$\overline{L_0}$ with values in $\overline{L_{-1}}$.
\end{proof}

This result enables us to recover in a more explicit way the ``equivalence'' of \cite{BaCra} and \cite{Cran} of a Lie 2-algebra to a skeletal one, and thus also to specify via our 3-cocycle their classification of Lie 2-algebras since skeletal Lie 2-algebras are classified by quadruples $(\mathfrak{g}, V, r, [c])$ with $\mathfrak{g}$ a Lie algebra,
$V$ a vector space, $r$ an action of $\mathfrak{g}$ on $V$ and $[c]$ the class in $H^3(\mathfrak{g}, V)$ of a 3-cocycle $c$ in $C^k(\mathfrak{g}, V)$. \par
On the other hand, F. Wagemann obtains in \cite{Wa1} and \cite{Wa2} an explicit classification of Lie algebra crossed-modules and proves that this classification is compatible with the classification of \cite{BaCra} and \cite{Cran}. \par 
More precisely, starting with a Lie algebra crossed module, seen as a strict Lie 2-algebra
$(L_{-1}\oplus L_0,l_1,l_2,l_3=0)$, F. Wagemann obtains a skeletal Lie 2-algebra $(\overline{L_{-1}}\oplus \overline{L_0}, \overline{l_1}=0, \overline{l_2}, \overline{l_3}^W)$ via the following diagram, where $\overline{l_2}$ is obtained in a unique way and $\overline{l_3}^W$ remains to be defined :  
$$\begin{matrix}       
&L_{-1}  &\stackrel{l_1}{\longrightarrow} & L_0&\\
&&&&\\
&\Bigg\uparrow \Phi_{1,-1} &  & \Bigg\uparrow \Phi_{1,0}&\\
&&&&\\
&\overline{L_{-1}}= \ker(l_1) &\stackrel{\overline{l_1}=0}{\longrightarrow}&\overline{L_0}= L_0/\mathrm{im} (l_1)&\\

\end{matrix}$$
\par 
Note that $\Phi_{1,-1}=\iota$ is the natural inclusion and the map $\Phi_{1,0}=\sigma$ is a linear section of the quotient map associating to an element $x$ in $L_0$ his class $\overline{x}$ in $\overline{L_0}$. The relation $(A_2)$ gives now explicitly $\Phi_2$ :
$$l_1(\Phi_2(\overline{x}, \overline{y}))=\sigma(\overline{l_2^p}(\overline{x}, \overline{y}))-l_2^p(\sigma(\overline{x}), \sigma(\overline{y}))$$
and if $\gamma$ denotes the 3-cocycle $\overline{l_3}^W$ obtained through this construction, the relation $(A_4)$ yields
$$\gamma=-(d_{\overline{L_0}})\Phi_2$$
where $d_{\overline{L_0}}$ is here a formal Chevalley-Eilenberg differential applied to the cochain $\Phi_2$, which has values in $L_{-1}$ instead of $\overline{L_{-1}}$. (Note that $L_{-1}$ is, in general, not a $\overline{L_{0}}$-module.)
\par
Via this construction, F. Wagemann associates to each Lie algebra crossed module a 3-cocycle $\gamma$, which is the key ingredient of his classification. 
\par
Let us link our construction explicitly to the classification of F. Wagemann via the following proposition. 
\begin{prop}
Let  a Lie algebra crossed module be given as a strict Lie 2-algebra $L=(L_{-1}\oplus L_0,l_1,l_2,l_3=0)$.
Denote by $\Phi=(\Phi_0,\Phi_1)$ the Lie 2-morphism of F. Wagemann [Wa2] described above and by 
$F=(F_1,F_2)$ the Lie 2-morphism described in Appendix A, and denote by $\gamma$
the 3-cocycle of F. Wagemann and by $\overline{l_3}$ the 3-cocycle associated to $\overline{L}$. Then the classes 
of $\gamma$ and $\overline{l_3}$ coincide in the $3$-th Lie algebra cohomology space $H^3(\overline{L_0},\overline{L_{-1}})$ corresponding to the Chevalley-Eilenberg differential of $\overline{L_0}$ with values in $\overline{L_{-1}}$.
\end{prop}

\begin{proof} 
First, we notice that in this case, since $\overline{L}$ is a strict Lie 2-algebra,  $\overline{l_3}$ is given as
$$\overline{l_3}= d_{{L_0}} F_2\, ,$$
where $d_{L_0}$ denotes  again a formal Chevalley-Eilenberg differential (here it corresponds to the Lie algebra $L_0$) of the cochain $F_2$ with values in $\overline{L_{-1}}$ (not in ${L_{-1}}$).
Let us consider the composition $F'=(F'_1,F'_2)$ of $F=(F_1,F_2)$ with
 $\Phi=(\Phi_0,\Phi_1)$, i.e., 
 $$F'=F\circ \Phi \, ,$$
where $F'_{1,-1} : \overline{L_{-1}}\mapsto \overline{L_{-1}}$,$F'_{1,0} :
\overline{L_0}\mapsto \overline{L_0}$ and 
$F'_2: \overline{L_0}\times \overline{L_0}\mapsto \overline{L_{-1}}$
 The relations $(C_1)$ and $(C_2)$ for a composition of Lie 2-algebra morphisms implies that $F'_{1,-1}$ and $F'_{1,0}$ are the identity of $\overline{L_{-1}}$ and $\overline{L_0}$ respectively, whereas the relation $(C_3)$ says that
 $$d_{\overline{L_0}} F_2' = \overline{l_3}-\gamma \, ,$$
 i.e., $\overline{l_3}$ and $\gamma$ differ by a coboundary in $B^3(\overline{L_0},\overline{L_{-1}}).$
 \end{proof}
We conclude that the 3-cocycle we associate (compare Appendix A) yields the same classification as the one found by J. Baez and A. Crans resp. by F. Wagemann.

\section{Multi Cartan calculus}

In this section the generalisation of some standard formulas of the calculus of vector fields and differential forms to multivector fields is recalled.\\

Let $M$ always be a smooth manifold. We will denote by $\mathfrak{X}^\bullet(M)$ the graded vector space given degree wise by $\mathfrak{X}^n(M)=\G(\wedge^n TM)$ and
by $\Omega^\bullet(M)$  the graded vector space given degree wise by
$\Omega^n(M)=\G(\wedge^nT^*M)$. 

\subsection{Multicontraction operator and multi Lie derivative}\hfill

If $X$ is in $\mathfrak{X}^1(M)$, the usual contraction $\iota_{X}$ is defined by
$$\iota_{X}(\a)=\a(X,..)$$
for $\a$ in $\Omega^\bullet(M)$ with the convention that $\iota_{X}(\a)=0$
if $\a$ is a function in $\Omega^0(M)=C^{\infty}(M)$. This contraction or inner product can be extended to a multicontraction operator for multivector fields in the following way.

\begin{defn}
 Let $v$ be a decomposable multivector field in $\mathfrak{X}^n(M)$, that is, $v=X_1\wedge \ldots\wedge X_n$ with $X_1,\ldots,X_n$ in $\mathfrak{X}^1(M)$. For $\alpha$ in $\Omega^{\bullet}(M)$, we define
 $$\iota_v(\a)=\iota_{X_n}\ldots \iota_{X_1}(\a)=\a(X_1,\ldots,X_n,...)$$
 with the convention that $\iota_v(\a)=0$ if $\a$ is in $\Omega^p(M)$ with $p<n$.
 \end{defn}
 
 The corresponding interior product $\iota^n:\mathfrak{X}^n(M)\times \Omega^p(M)\mapsto \Omega^{p-n}(M)$ (for $p\geq n$)
 is then extended to a well-defined operator $\iota_v$ by $C^{\infty}(M)$-linearity for any multivector field $v$, even if $v$ is not a decomposable
 multivector field.
 
 It is well-known that, for any $X$ in $\mathfrak{X}^1(M)$, the usual Lie derivative $\mathcal{L}_X$ satisfies Cartan's magic formula 
 $$\mathfrak{L}_{X}=d\circ \iota_{X} +\iota_{X}\circ d.$$
   The notion of Lie derivative can be extended to multivector fields
  just  by requiring the analogue of the Cartan's magic rule.
   \begin{defn}
       Let $\a$ be in $\Omega^{\bullet}(M)$. For any multivector field $v$ in $\mathfrak{X}^{\bullet}(M)$, we define the multi Lie derivative $\mathcal{L}_v$ by 
       $$\mathcal{L}_{v}(\a)= d \iota_v(\a)-(-1)^{|v|} \iota_v d(\a)$$
       where $|v|=n$ if $v$ is in $\mathfrak{X}^n(M)$.
       
   \end{defn}

 \subsection{Multi Cartan commutation rules}
 
Multi Cartan commutation rules can be written in complete generality for multivector fields. 
In our case, we will have to use only the following proposition.
 \begin{prop}
     Let $X$ and $Y$ be in $\mathfrak{X}^1(M)$ and let $v$ be in $\mathfrak{X}^2(M)$, and denote by $[\,,\,]_S$ the usual Schouten bracket. We have
     \begin{itemize}
     \item $\iota_{X}\circ \iota_{Y}+\iota_{Y}\circ \iota_{X}=0$
     \item $\iota_{X}\circ  \iota_{v}-\iota_{v}\circ \iota_{X}=0$
     \item $\mathfrak{L}_{X}\circ \iota_{Y} -\iota_{Y} \circ \mathfrak{L}_{X}=\iota_{[X,Y]_S}$
     \item $\mathfrak{L}_{X}\circ \iota_{v} -\iota_{v} \circ \mathfrak{L}_{X}=\iota_{[X,v]_S}$
     \item $\mathfrak{L}_{X}\circ \mathfrak{L}_{Y} -
     \mathfrak{L}_{Y}\circ \mathfrak{L}_{X}=\mathfrak{L}_{[X,Y]_S}$
     \item  $\mathfrak{L}_{v} \circ \mathfrak{L}_X-\mathfrak{L}_{X} \circ \mathfrak{L}_v = \mathfrak{L}_{[v,X]_S}$
     \item 
     $ 
         \mathfrak{L}_{X\wedge Y}  = 
     \mathfrak{L}_Y\circ \iota_X-\iota_Y\circ\mathfrak{L}_X \\
      {\color{white}\mathfrak{L}_{X\wedge Y} } =  \iota_X\circ\mathfrak{L}_Y-\iota_Y\circ\mathfrak{L}_X-\iota_{[X,Y]_S} \\
    {\color{white} \mathfrak{L}_{X\wedge Y} } =  \mathfrak{L}_Y\circ \iota_X-\mathfrak{L}_X\circ \iota_Y+\iota_{[X,Y]_S}.$ 
    \end{itemize}
 \end{prop}

\begin{rema}
There is obviously a choice in the definition of the multicontraction $\iota$ concerning the order in which the components are inserted into the form. In this paper, we use the same convention as in \cite{Rog} and \cite{Mi}. It is important to note that N.L. Delgado \cite{Del} uses a different convention.
\end{rema}

\section{Action of a Lie 2-algebra (2-action)}

Following N. Delgado, a Lie 2-algebra structure is given on the sum of the spaces of 2-vector fields and of vector fields on a manifold $M$, as well as the notion of a 2-action $\rho$ of this algebra on $M$. We observe that this is not a "homotopy invariant" notion, i.e., a quasi-isomorphism of Lie 2-algebras does not preserve 2-actions.\\

Recall that, in the usual setting, an (infinitesimal) action of a Lie algebra $\mathfrak{g}$ on a manifold $M$ is given by a Lie algebra morphism 
$\mathfrak{g}\longrightarrow\mathfrak{X}^1(M)$. We now want to extend this definition to the case when $\mathfrak{g}=\mathfrak{g}_{-1}\oplus\mathfrak{g}_{0}$ a Lie 2-algebra.\\ 

We consider the following Lie 2-algebra structure on vector fields, including the 2-vector fields (compare \cite{Del}) :

\begin{prop}{(Structure of Lie 2-algebra on $\mathfrak{X}^2(M)\oplus\mathfrak{X}^1(M)$)}
 
Let $M$ be a manifold, $\displaystyle (\mathfrak{X}^2(M)\oplus\mathfrak{X}^1(M), \nu_1, \nu_2, \nu_3)$ is a Lie 2-algebra with the following brackets : 
 \begin{enumerate}
     \item $\nu_1: \mathfrak{X}^2(M)\rightarrow\mathfrak{X}^1(M)$ is the zero map
     \item $\nu_2$ is the Schouten bracket, more precisely
     \begin{enumerate}
     \item ${\nu_2^p}: \mathfrak{X}^1(M)\times\mathfrak{X}^1(M) \rightarrow\mathfrak{X}^1(M)$ is antisymmetric and ${\nu_2^p}(X,Y)=[X ,Y]_S$ is the usual bracket on vector fields (i.e., the Schouten bracket on $\mathfrak{X}^1(M)$)
     \item ${\nu_2^m}:\mathfrak{X}^2(M)\times\mathfrak{X}^1(M)\rightarrow \mathfrak{X}^2(M)$ is given by ${\nu_2^m}(Z,X\wedge Y)=[Z,X\wedge Y]_S=[Z,X]_S\wedge Y+X\wedge[Z,Y]_S$, with ${\nu_2^m}(Z,X\wedge Y)=-{\nu_2^m}(X\wedge Y,Z)$.

     Simplifying the notations, we will often write $\nu_2=[\ , \ ]_S$ instead of $\nu_2^p$ and $\nu_2^m$.
     \end{enumerate}
     \item $\nu_3: \bigwedge^3(\mathfrak{X}^1(M))\rightarrow \mathfrak{X}^2(M)$ is the zero map.
 \end{enumerate}
 \end{prop}
\begin{proof}
It easily follows that the above brackets indeed define a Lie 2-algebra.
The relations $(R_1), (R_2)$ and $(R_3)$ are satisfied since $\nu_1=0$, and $(R_6)$ is obvious since $\nu_3=0$. The relations $(R_4)$ and $(R_5)$ are, respectively, the Jacobi identities on $\mathfrak{X}^1(M)$ and on $\mathfrak{X}^2(M)\oplus\mathfrak{X}^1(M)$ .
\end{proof}

\begin{rema}\label{multifields}
This construction can be generalised to a Lie $n$-algebra structure on $\mathfrak{X}^n(M)\oplus \ldots \oplus \mathfrak{X}^2(M)\oplus\mathfrak{X}^1(M)$ (all brackets vanishing except $\nu_2$ being the Schouten bracket).
\end{rema}

Now, specialising \cite{Del}, we give the crucial

\begin{defn}{(2-action)}

An \emph{(infinitesimal) action of a Lie 2-algebra} $\mathfrak{g}=(\mathfrak{g}_{-1}\oplus\mathfrak{g}_{0}, l_1, l_2, l_3)$ on a manifold $M$ (or in short a \emph{2-action} of $\mathfrak{g}$ on $M$) is a Lie 2-morphism   
$$\rho: \mathfrak{g}_{-1}\oplus\mathfrak{g}_{0} \longrightarrow \mathfrak{X}^2(M)\oplus\mathfrak{X}^1(M) .$$

\medskip

i.e.
using the following decomposition of $\rho$,
\begin{itemize}
\item[(i)]$\rho_{1,0}:\mathfrak{g}_0 \longrightarrow \mathfrak{X}^1(M)$
\item[(ii)] $\rho_{1,-1}:\mathfrak{g}_{-1} \longrightarrow \mathfrak{X}^2(M)$
\item[(iii)] $\rho_{2}:\bigwedge^2\mathfrak{g}_0 \longrightarrow \mathfrak{X}^2(M)$
\end{itemize}
one has, $\forall a\in \mathfrak{g}_{-1}$, $\forall x, y, z\in \mathfrak{g}_0$, the following identities
 \begin{itemize}
     \item[$(A_1)$] $\displaystyle\rho_{1,0}(l_1(a))=0$, 
     \item[$(A_2)$] $\rho_{1,0}(l_2^p(x,y))=
     [\rho_{1,0}(x),\rho_{1,0}(y)]_S$
     \item[$(A_3)$] $\rho_{1,-1}(l_2^m(a,x))=
     [\rho_{1,-1}(a),\rho_{1,0}(x)]_S+ \rho_2(l_1(a),x)$
     \item[$(A_4)$] $\big{(}\rho_2(l_2^p(x,y),z)+c.p.\big{)}+ \rho_{1,-1}(l_3(x,y,z))=
     [\rho_{1,0}(x),\rho_2(y,z)]_S+ c.p.$
     \end{itemize}
\end{defn}
 \begin{rema}\hfill

 \begin{enumerate}
 \item The given identities follow, of course, from the definition of the Lie 2-algebra structure on $\mathfrak{X}^2(M)\oplus\mathfrak{X}^1(M)$ and the definition of a Lie 2-algebra morphism.
     \item The relation $(A_2)$ says that $\rho_{1,0}$ preserves the brackets even if $\mathfrak{g}_0$ is not a Lie algebra.
     \item  Choosing $\rho_{1,-1}=0$, $\rho_2=0$ and $\rho_{1,0}$ satisfying $(A_1)$ and $(A_2)$
     corresponds to a Lie algebra action of (the Lie algebra) $\mathfrak{g}/\mathrm{im}\, l_1$ on $M$.
     \item Following Remark \ref{multifields} this definition can be generalised to a definition of a \emph{n-action} as a Lie $n$-algebra morphism : $\displaystyle\bigoplus_{k=1-n}^{0}\mathfrak{g}_k \longrightarrow \bigoplus_{k=n}^{1}\mathfrak{X}^k(M)$.
     
 \end{enumerate}
\end{rema}

\begin{defn}
We say that the 2-action is \emph{strict} if $\rho_2=0$. This corresponds to the conditions $(A_1)$ and $(A_2)$ above, and simplifies the other two conditions to

$(A_3)$  $\rho_{1,-1}(l_2^m(a,x))=
     [\rho_{1,-1}(a),\rho_{1,0}(x)]_S$ resp. $(A_4)$  $\rho_{1,-1}(l_3(x,y,z))=0.$   
\end{defn}
\begin{rema}
     If the 2-action is strict and $\mathfrak{g}_{-1}$ is ``small'', one has $\rho_{1,-1}=0$, for example in the case $\mathfrak{g}_{-1}=l_3(\bigwedge^3\mathfrak{g}_0)$.
     
 \end{rema}
The notion of an infinitesimal 2-action generalises the notion of an action of a Lie algebra :

\begin{prop} 
    A Lie algebra action of a Lie algebra $\mathfrak{g}_{0}$ on a manifold $M$ corresponds to a strict action of the Lie 2-algebra $\{0\}\bigoplus\mathfrak{g}_{0}$ on $M$.
\end{prop}

\begin{proof}
The result
follows directly from the definitions.
\end{proof}

\begin{defn}
Let $(\mathfrak{g}, \mathfrak{h}, \tau, r)$ be a Lie algebra crossed module (which can be seen as a strict Lie 2-algebra). We say that $\rho$ is a \emph{crossed module action} of  the Lie algebra crossed module $(\mathfrak{g}, \mathfrak{h}, \tau, r)$ on a manifold $M$ if $\rho$ is a strict Lie 2-action. In other words, a crossed module action of a Lie algebra crossed module $(\mathfrak{g}, \mathfrak{h}, \tau, r)$
on a manifold $M$ is a crossed module morphism between $(\mathfrak{g}, \mathfrak{h}, \tau, r)$ and $\mathfrak{X}^2(M)\oplus\mathfrak{X}^1(M)$ (viewed as a Lie algebra crossed module).
\end{defn}

 \begin{prop}
     The quasi-isomorphism of Proposition \ref{quasi-iso} does not preserve 2-actions. 
 \end{prop}

\begin{proof}
Let $\left(\mathfrak{g}=\mathfrak{g}_{-1}\oplus\mathfrak{g}_{0}, l_1, l_2, l_3\right)$ be a Lie 2-algebra and denote by $\left(\overline{\mathfrak{g}}=\overline{\mathfrak{g}_{-1}}\oplus\overline{\mathfrak{g}_{0}}, \overline{l_1}, \overline{l_2}, \overline{l_3}\right)$ the associated skeletal Lie 2-algebra, obtained through the quasi-isomorphism in Appendix A. We denote by $F=(F_1,F_2)$ this quasi-isomorphism.
Assume there is a 2-action $\overline{\rho}$ such that the following diagram commutes

\[
\begin{tikzpicture}[>=latex, node distance=2.5cm, auto]
  \node (A) at (0,0) {$\left(\overline{\mathfrak{g}}=\overline{\mathfrak{g}_{-1}}\oplus\overline{\mathfrak{g}_{0}}, \overline{l_1}, \overline{l_2}, \overline{l_3}\right)$};
  \node (B) at (5,3) {$\mathfrak{X}^2(M)\oplus\mathfrak{X}^1(M)$};
  \node (C) at (0,3) {$\left(\mathfrak{g}=\mathfrak{g}_{-1}\oplus\mathfrak{g}_{0}, l_1, l_2, l_3\right)$};
  \draw[->] (A) -- node[above left] {$\overline{\rho}$} (B);
  \draw[->] (C) -- node[above right] {$\rho$}(B);
  \draw[->] (C) -- node[above left] {$F$}(A);
\end{tikzpicture}
\]

                           

Then the following relations must hold: 
\begin{itemize}
\item [$(C_1)$]  $\rho_{1,0}=\overline{\rho_{1,0}}\circ F_{1,0}$
   \item [$(C_2)$]  $\rho_{1,-1}=\overline{\rho_{1,-1}}\circ F_{1,-1}$
   \item [$(C_3)$] $\rho_2=\overline{\rho_2}\circ (F_{1,0}\times F_{1,0})+\overline{\rho_{1,-1}}\circ F_2.$
\end{itemize}
If $\overline{{\rho}_{1,-1}}{\vert_{\overline{\mathfrak{g}_{-1}}}}= 0 $ and ${\rho_{1,-1}}{\vert_{\ker l_1}} = 0 $, but ${\rho_{1,-1}}{\vert_ {\mathfrak{g}_{-1}}}\neq 0 $, relation $(C_2)$ fails. (See Appendix C, Example 5b, Remark \ref{quasiisomorphism}.) 
\end{proof}

\begin{rema} 
Using a quasi-inverse $G$ such that 
$F\circ G\simeq \mathrm{id}_\mathfrak{g}$ and 
$G\circ F\simeq \mathrm{id}_{\overline{\mathfrak{g}}}$ 
in the sense of Appendix A of \cite{LaLi2}, one has, of course, 2-actions 
$\bar{\rho}=\rho\circ G$ of $\overline{\mathfrak{g}}$, resp. $\tilde{\rho}=\bar{\rho}\circ F$ of $\mathfrak{g}$. 
Furthermore, the actions $\rho$ and $\tilde{\rho}$ are homotopic as  Lie 2-algebra morphisms. (The existence of $G$ is shown, e.g., in \cite{KrSc}, Lemma 6.12, upon observing that Lie 2-algebras form a full subcategory of $L_\infty$-algebras).
\end{rema}

\begin{rema} 
In the classical case, when $\mathfrak{g}_{0}$ is a Lie algebra and $M$ a smooth manifold, an (infinitesimal) action of $\mathfrak{g}_{0}$ on $M$ is a Lie algebra morphism $\mathfrak{g}_{0} \longrightarrow \mathfrak{X}^1(M)$. Since $\mathfrak{X}^1(M) \cong \mathsf{Der}(C^\infty(M))\subset\End(C^\infty(M))$, it is a Lie algebra morphism from $\mathfrak{g}_{0}$ to the Lie algebra $\End(C^\infty(M))$ with its usual Lie bracket coming from its associative multiplication. If $(M,\omega)$ is a symplectic manifold, $C^\infty(M)$ can be viewed as the space of observables.

Thus, for a 2-plectic manifold $(M,\omega)$ it is natural to interpret an (infinitesimal) 2-action $\rho$ as a morphism of Lie 2-algebras from $\mathfrak{g}=\mathfrak{g}_{-1}\oplus\mathfrak{g}_{0}$ to the endomorphism space of its ``observables''. 

This latter space of of observables on a 2-plectic $M$ should be of the form of an appropriate graded vector space $V=V_{-1} \oplus V_0$, its endomorphisms canonically carrying the structure of a Lie 2-algebra (see, e.g., \cite{LaLi2}).

 Upon naively choosing the observables as a subspace of $C^\infty(M)\oplus\Omega^1(M)$, we are naturally led to consider
 Lie 2-algebra morphisms from $\mathfrak{X}^2(M)\oplus\mathfrak{X}^1(M)$ to $\End(C^\infty(M)\oplus\Omega^1(M))$. Such morphisms exist and will be discussed in Appendix \ref{endo}.
\end{rema}

 
 \section{Lie 2-algebras of observables} 

 We define observable algebras on a 2-plectic manifold $(M,\o)$, in the sense of Baez and Rogers and in the sense of Delgado, and relate them mutually and to multivector fields on $M$.
 
\subsection{Hamiltonian vector fields and forms}

\begin{defn}
Given a $n$-plectic manifold $(M,\omega)$, a
  multivector field $v$ in $\mathfrak{X}^{\bullet}(M)$ is called \emph{multisymplectic} if $\mathcal{L}_v(\omega)=0$. We denote by $\mathfrak{X}^{\bullet}_{Sympl}(M)$ the set of multisymplectic multivector fields.
\end{defn} 

\begin{rema}
Let $(M,\o)$ be a $n$-plectic manifold.
    A multivector field $v$  in
    $\mathfrak{X}^{\bullet}(M)$ is a multisymplectic multivector field if and only if $\iota_v(\omega)$ is closed.
\end{rema} 

Let us  refine the notion of 2-action in the 2-plectic case. (Of course, this might easily be generalised to the $n$-plectic situation).

\begin{defn}
Let $(M,\omega)$ be a 2-plectic manifold.  A 2-action $\rho : \mathfrak{g}_{-1}\oplus\mathfrak{g}_{0} \longrightarrow \mathfrak{X}^2(M)\oplus\mathfrak{X}^1(M)$ is called \emph{2-plectic} if for all $x,y\in 
\mathfrak{g}_{0}$ and for all $a\in \mathfrak{g}_{-1}$ one has
$$\mathcal{L}_{\rho_{1,0}(x)}\omega=0, \,\, 
\mathcal{L}_{\rho_{1,-1}(a)}\omega=0, \,\, \hbox{and} \,\,  
\mathcal{L}_{\rho_{2}(x,y)}\omega=0 \, .$$
If only the first two Lie derivatives vanish,  then $\rho$ is called \emph{quasi-2-plectic}.
\end{defn} 
\begin{defn}
Given a $n$-plectic manifold $(M,\o)$,
a $(n{-}1)$-form $\alpha$  is \emph{Hamiltonian} if there exists a vector field $X_\alpha \in\mathfrak{X}^1(M)$ such that
$d\alpha =- \iota_{X_\alpha} \omega.$
The vector field $X_\alpha$ is called the \emph{Hamiltonian vector field corresponding to $\alpha$}.
We denote by $\Omega^{n-1}_{Ham}(M)$ the set of Hamiltonian $(n-1)$-forms on $M$.
\end{defn}

\subsection{Rogers's Lie 2-algebra of observables} \hfill

\bigskip

We recall the Lie n-algebra $\mathsf{L^n(M,\omega)}$ of J. Baez and C. Rogers associated with an $n$-plectic manifold and specialise its brackets in the case $n=2$.

\begin{thm}[\cite{Rog}] Given a $n$-plectic manifold, there is an associated Lie $n$-algebra, denoted $\mathsf{L^n(M,\omega)}$ and called the Lie $n$-algebra of observables.
\end{thm} 

Let us give an explicit description of the statement of the preceding theorem in the case that $(M,\omega)$ is a 2-plectic manifold. 
\begin{prop}[\cite{Rog}]
     Let $(M,\omega)$ be a 2-plectic manifold. The graded vector space
     $$\mathsf{L^2(M,\omega)}=C^{\infty}(M)\oplus \Omega^1_{Ham}(M),$$
     where the elements of $C^{\infty}(M)$ are of degree -1 and the elements of $\Omega^1_{Ham}(M)$ are of degree 0, together with the following brackets :
\begin{itemize}
\item $\Tilde{l_1^R} : C^{\infty}(M)\mapsto \Omega_{Ham}^1(M)$,  
$\Tilde{l_1^R}(f)=df$ 
\item  $\Tilde{l_2^R} :  \bigwedge^2( \Omega^1_{Ham}(M))\longrightarrow 
\Omega_{Ham}^1(M)$, 
$\Tilde{l_2^R} (\a,\b)(\omega)= \iota_{X_{\a}\wedge X_{\b}}(\omega)=\iota_{X_{\b}}\iota_{X_{\a}}(\omega)$
where $$d\iota_{X_{\a}\wedge X_{\b}}(\omega)=-\iota_{[X_\alpha,X_\beta]_S}(\omega)$$
\item $\Tilde{l_3^R}: \bigwedge^3( \Omega^1_{Ham}(M)) \longrightarrow C^{\infty}(M)$, $\Tilde{l_3^R}(\alpha,\beta,\gamma)=-\iota_{{X_\alpha}\wedge{X_\beta}\wedge{X_\gamma}}(\omega)$.
\end{itemize}
defines a Lie 2-algebra  $(\mathsf{L^2(M,\omega)}, \Tilde{l_1^R}, \Tilde{l_2^R}, \Tilde{l_3^R})$, sometimes called the Baez-Rogers algebra of observables.
 \end{prop}
 
 \subsection{Delgado's Lie 2-algebra of observables}\hfill

 \medskip

We define Delgado's Lie 2-algebras of observables $\mathsf{D^2(M,\omega)}$ and 
$\mathsf{\widetilde{D^2(M,\omega)}}$, and introduce the important 2-morphisms $\Phi : \mathsf{\widetilde{D^2(M,\omega)}} \to \mathsf{L^2(M,\omega)}$ and 
$\Psi : \mathsf{\widetilde{D^2(M,\omega)}} \to \mathfrak{X}^2(M)\oplus\mathfrak{X}^1(M)$.\\

The following definition will be useful in the sequel. 
  \begin{defn}
Let  $(M,\omega)$ be a 2-plectic manifold. 
 A function $f$ in $C^{\infty}(M)$ is said to be \emph{Hamiltonian (or in $C^{\infty}_{Ham}(M)$)} 
if there exists  $v_f$ in $\mathfrak{X}^2(M)$ such that $df=-\iota_{v_f}\omega$.
\end{defn}

\begin{rema}
    There can be two (or more) distinct $v_f$ and $w_f$ in $\mathfrak{X}^2(M)$ such that $df=-\iota_{v_f}\omega = -\iota_{w_f}\omega$ since the two-kernel of $\o$,
$$\ker \o^{(2)}=\{ u\in \mathfrak{X}^2(M) \,  \vert \, \iota_u\o = 0 \}\, ,$$
is typically nontrivial. 
 
\end{rema}

\medskip 

\begin{example}

Let $M=\mathbb R^5= \mathbb{R} \times \mathbb{R}^4$ with coordinates $(t,q^1,q^2,p_1,p_2)$ and 
$\omega=dt\wedge dq^1\wedge dp_1 + dt\wedge dq^2\wedge dp_2$. Then
contraction with $\omega$ from $\Lambda^2TM$ to $T^*M$ is surjective and
$\ker \o^{(2)}$ has the following base in $m\in M$:
$$\{\partial_{q^1}\wedge \partial_{q^2}, \partial_{p_1}\wedge \partial_{p_2}, \partial_{q^1}\wedge \partial_{p_2}, \partial_{q^2}\wedge \partial_{p_1}, 
\partial_{q^1}\wedge \partial_{p_1}-\partial_{q^2}\wedge \partial_{p_2}\}\, ,$$
where we note here succintly $\partial_{q^1}$ for $\frac{\partial } {\partial{q^1}}\big{\vert}_m$ etc.\\

\end{example}

Following N.L Delgado \cite{Del} we first introduce $\mathsf{D^2(M,\omega)}$,
an extension of the Lie 2-algebra of observables $\mathsf{L^2(M,\omega)}$. (For the convenience of the reader we verify the key properties below.) Moreover, due to the previous remark, we will also introduce in the sequel the Lie 2-algebra of observables $\mathsf{\widetilde{D^2(M,\omega)}}$ that includes Hamiltonian pairs. 
 
\begin{prop}
Let $(M,\omega)$ be a 2-plectic manifold. Then the graded vector space
$$\mathsf{D^2(M,\omega)}=(C^{\infty}(M)\times C^{\infty}_{Ham}(M))\oplus  \Omega^1_{Ham}(M),$$ where the elements of $C^{\infty}(M)\times C^{\infty}_{Ham}(M)$ are of degree -1 and the elements of $\Omega^1_{Ham}(M)$ are of degree 0, together with the following brackets 
\begin{itemize}
\item $\Tilde{l_1} : C^{\infty}(M)\times C^{\infty}_{Ham}(M)\mapsto \Omega_{Ham}^1(M)$,  
$\Tilde{l_1}(\Tilde{f},f)=d\circ pr_1(\Tilde{f},f)= d\Tilde{f}$, where $pr_1(\Tilde{f},f)=\Tilde{f}$
\item $\Tilde{l_2}$ which is decomposed into two antisymmetric maps $\Tilde{l_2^p}$
and $\Tilde{l_2^m}$ given by\\
$\Tilde{l_2^p }: \bigwedge^2(\Omega^1_{Ham}(M))\longrightarrow  \Omega^1_{Ham}(M), \, \, 
\Tilde{l_2^p} (\a,\b)= \iota_{X_{\a}\wedge X_{\b}}(\omega)\, ,$ where
$d\iota_{X_{\a}\wedge X_{\b}}(\omega)=-\iota_{[X_\alpha,X_\beta]_S}(\omega)$
and
$\Tilde{l_2^m} : (C^{\infty}(M)\times C^{\infty}_{Ham}(M))\times \Omega^1_{Ham}(M) \longrightarrow C^{\infty}(M)\times C^{\infty}_{Ham}(M),$ 
$$\Tilde{l_2^m}((\Tilde{f}, f),\alpha)=(0,\iota_{v\wedge X_\alpha}(\omega)), \hbox{where } v\hbox{ is such that } df=-\iota_v\omega$$ 

The definition of $\Tilde{l_2^m}$ is independent of the choice of $v$ and $\Tilde{l_2^m}$ is  antisymmetric upon setting
$\Tilde{l_2^m}(\a, (\Tilde{f},f))=-(0,\iota_{v\wedge X_\alpha}(\omega))=
-(0,\iota_{X_{\alpha}\wedge v}(\omega)).$
\item $\Tilde{l_3}: \bigwedge^3( \Omega^1_{Ham}(M)) \longrightarrow C^{\infty}(M)\times C^{\infty}_{Ham}(M)$, $\Tilde{l_3}(\alpha,\beta,\gamma)=
(-\iota_{X_\alpha\wedge X_\beta\wedge X_\gamma}(\omega), 0)$
\end{itemize} 
defines a Lie 2-algebra $(\mathsf{D^2(M,\omega)}, \Tilde{l_1}, \Tilde{l_2}, \Tilde{l_3})$.
\end{prop}

\begin{proof} 
 Let us first prove that 
$$\Tilde{l_2^m}((\Tilde{f}, f),\alpha)=(0,\iota_{v\wedge X_\alpha}(\omega))$$
does not depend on the choice of $v$. Let $v$ and $w$ such that
$$df=-\iota_v\omega=-\iota_w\omega.$$ We observe that
$$\iota_{v\wedge X_\alpha}(\omega)=\iota_{X\alpha}\iota_v(\omega)=
\iota_{X_{\alpha}}\iota_w(\omega)=\iota_{w\wedge X_{\alpha}}(\omega).$$
Now, to verify that the brackets $\tilde{l_k}$ are well-defined, we also have to check that
    $\mathrm{im} (pr_2 \circ\Tilde{l_2^m})$ is in $C^{\infty}_{Ham}(M)$. For all 
 $v$ in $\mathfrak{X}^2(M)$ (associated to $f$ in $C^{\infty}_{Ham}(M)$), and $\a$ in $\Omega_{Ham}^1(M)$, we see that
$$d (\iota_{v\wedge X_{\a}}(\o))=d(\iota_{X_{\a}} \iota_v(\o))= \mathfrak{L}_{X_{\a}}\iota_v(\o)$$
since $\mathfrak{L}_{X_{\a}}= d\circ \iota_{X_{\a}}+\iota_{X_{\a}}\circ d$ and $d\iota_v(\o)=0$.
Thus
$$d(\iota_{v\wedge X_{\a}}(\o))=\iota_{[X_{\a},v]_S}(\o)=-\iota_{[v,X_{\a}]_S}(\omega)$$ since
$$\mathfrak{L}_{X_{\a}}\circ \iota_v-\iota_v\circ \mathfrak{L}_{X_{\a}}=\iota_{[X_{\a},v]_S}$$ and $\mathfrak{L}_{X_{\a}}(\o)=0.$
Now, we have to show that the relations $(R_i)$ are all satisfied. 
Relation $(R_1)$ clearly being fulfilled, we begin here with relation $(R_2)$. We have for all $\a$ in $\Omega_{Ham}^1(M)$ and for all $(\Tilde{f},f)$ in $C^{\infty}(M)\times C^{\infty}_{Ham}(M),$
$$\Tilde{l_1} \Tilde{l_2^m} (\a,(\Tilde{f},f))=-\Tilde{l_1}(0, \iota_{v\wedge X_{\a}}(\omega))=0\, ,$$
where $v$ satisfies $df=-\iota_v(\omega)$,
as well as
$$\Tilde{l_2^p}(\a, l_1(\Tilde{f},f))=\Tilde{l_2^p}(\a,d\Tilde{f})=0.$$ 

Let us now check relation $(R_3)$. We have for all 
$(\Tilde{f},f)$ and $(\Tilde{g},g)$ in $C^{\infty}(M)\times C^{\infty}_{Ham}(M)$
with $v$ and $v'$ in $\mathfrak{X}^2(M)$ such that 
$df=-\iota_v(\omega)$ and $dg=-\iota_{v'}(\omega),$
$$\Tilde{l_2^m}(l_1(\Tilde{f},f), (\Tilde{g},g))=\Tilde{l_2^m}(d\Tilde{f},(\Tilde{g},g))=(0,0)\, .$$

By antisymmetry, one has as well
$$\Tilde{l_2^m}((\Tilde{f},f), l_1(\Tilde{g},g))=(0,0).$$
Moreover, since $\Tilde{l_2^p}$ coincides with $\Tilde{l_2^R}$ and the term $\Tilde{l_1}(\Tilde{l_3}(\alpha,\beta,\gamma))$ being equal to 
$$\Tilde{l_1^R}(\Tilde{l_3^R}(\alpha,\beta,\gamma)),$$
we observe that the relation $(R_4)$ for $\mathsf{D^2(M,\omega)}$ is satisfied just as the relation $(R_4)$ for $\mathsf{L^2(M,\omega)}$ (see \cite{Rog}) is satisfied. This relation says that
$$d(\iota_{X_\alpha\wedge X_\beta\wedge X_\gamma}(\omega))= \iota_Y(\omega)$$ where 
$Y=[X_\alpha,X_\beta]_S\wedge X_\gamma + c.p.$, and it is  obtained by
using the properties of Proposition 4.3.

Let us prove the relation $(R_5)$. 
On the one hand, we have
$$\Tilde{l_3}(\Tilde{l_1}(\Tilde{f},f),\a,\b)=\Tilde{l_3}(d\Tilde{f},\a,\b)=(0,0).$$
On the other hand, we have
$$\Tilde{l_2^m}(\Tilde{l_2^p}(\a,\b),(\Tilde{f},f))+
\Tilde{l_2^m}(\Tilde{l_2^m}(\b,(\Tilde{f},f)),\a)+
\Tilde{l_2^m}(\Tilde{l_2^m}((\Tilde{f},f),\a),\b)$$
$$ =(0, -\iota_v\iota_{[X_{\a},X_{\b}]_S}(\o)- \iota_{[v, X_{\b}]_S}\iota_{X_\a}(\o) +\iota_{[v,X_{\a}]_S} \iota_{X_\b}(\o))$$
$$ = (0,-\iota_{[X_\a,X_\b]_S} \iota_v(\o) +\iota_{X_\a} \iota_{[X_\b,v]_S}(\o)
-\iota_{X_\b}\iota_{[X_\a,v]_S}(\o)$$
where $v$ satisfies $df=-\iota_v(\omega).$
Moreover, we can observe that
$$\iota_{X_\a}\iota_{[X_\b,v]_S}(\o)= \iota_{X_\a} d(\iota_{X_\b}\iota_v(\o)) $$
$$= \mathfrak{L}_{X_{\a}}(\iota_{X_\b}\iota_v(\o)) = -\mathfrak{L}_{X_\a}\iota_{X_\b}(df)= -\mathfrak{L}_{X_\a} \mathfrak{L}_{X_\b}(f)$$
since one can check that $\iota_{X_{\a}}\iota_{X_{\b}}\iota_v(\o)=0$ 
and since $\iota_{X_{\b}}(f)=0.$
In the same way,
$$\iota_{X_\b}\iota_{[X_\a,v]_S}(\o)=-\mathfrak{L}_{X_\b} \mathfrak{L}_{X_\a} (f)$$
and $$\iota_{[X_\a,X_\b]} \iota_v(\o)=-\mathfrak{L}_{[X_{\a},X_{\b}]_S}(f).$$
We have thus proved that 
$$\Tilde{l_2^m}(\Tilde{l_2^p}(\a,\b),(\Tilde{f},f))+
\Tilde{l_2^m}(\Tilde{l_2^m}(\b,(\Tilde{f},f)),\a)+
\Tilde{l_2^m}(\Tilde{l_2^m}((\Tilde{f},f),\a),\b)=(0,0)$$
since 
$$\mathfrak{L}_{X_{\a}}\circ \mathfrak{L}_{X_{\b}}-\mathfrak{L}_{X_{\b}}\circ \mathfrak{L}_{X_{\a}}=\mathfrak{L}_{[X_{\a},X_{\b}]_S}\, .$$ 
Finally, we have to prove $(R_6)$. We first observe that for all 
$\alpha$, $\beta$, $\gamma$, $\delta$ in $\Omega^1_{Ham}(M)$, we have
$$\Tilde{l_3}(\Tilde{l_2}(\alpha,\beta),\gamma,\delta)-\Tilde{l_3}(\Tilde{l_2}(\beta,\gamma),\alpha,\delta)+ \Tilde{l_3}(\Tilde{l_2}(\alpha,\delta),\beta,z)+\Tilde{l_3}(\Tilde{l_2}(\beta,z),\alpha,\delta)-\Tilde{l_3}(\Tilde{l_2}(\beta,\delta),\alpha,\gamma)=$$ 
$$\bigg(\Tilde{l_3^R}(\Tilde{l_2^R}(\alpha,\beta),\gamma,t)-\Tilde{l_3^R}(\Tilde{l_2^R}(\alpha,\gamma),\beta,\delta)+ \Tilde{l_3^R}(\Tilde{l_2^R}(\alpha,\delta),\beta,\gamma)+\Tilde{l_3^R}(\Tilde{l_2^R}(\beta,\gamma),\alpha,\delta)-\Tilde{l_3^R}(\Tilde{l_2^R}(\beta,\delta),\alpha,\gamma),0\bigg)$$
which is equal to $(0,0)$ as in the case of $\mathsf{L^2(M,\omega)}$ (\cite{Rog}) and using the fact
that 
$$d\omega(X_{\alpha},X_{\beta},X_{\gamma},X_{\delta})=0.$$
By the definition of $\Tilde{l_2^m}$ and $\Tilde{l_3}$, we have 
 $$\Tilde{l_2^m}(\Tilde{l_3}(\alpha,\beta,\gamma),\delta)-\Tilde{l_2^m}(\Tilde{l_3}(\alpha,\beta,\delta),\gamma)+\Tilde{l_2^m}(\Tilde{l_3}(\alpha,\gamma,\delta),\beta)-\Tilde{l_2^m}(\Tilde{l_3}(\beta,\gamma,\delta),\alpha)=(0,0),$$
 completing the proof.
\end{proof}

\begin{rema}\hfill

\begin{enumerate}
    \item Comparing Rogers's and Delgado's observables, we emphasize the presence of mixed brackets in $\mathsf{D^2(M,\omega)}$ whereas Rogers's brackets are non zero only on $1-$forms.
    \item We underline the important relation $$d(\iota_{v\wedge X_{\a}}(\o))=d(\iota_{v}\iota_{X_\alpha}\omega)=-\iota_{[v,X_\alpha]_S}\omega=\iota_{[X_\alpha,v]_S}\omega \, ,$$
which is useful in the preceding proof to show that $\Tilde{l_2^m}$ is well-defined. \\
We have the similar following relation 
$\displaystyle\ d(\iota_{X_\alpha\wedge X_\beta\wedge X_\gamma}\omega)= \iota_Y\omega$ where 
$Y=[X_\alpha,X_\beta]_S\wedge X_\gamma + c.p.$, which is the relation $(R_4)$ both for 
$\mathsf{L^2(M,\omega)}$ and $\mathsf{D^2(M,\omega)}$. 
\item It is also interesting to see that
$$\Tilde{l_2^m}((0,-\iota_{X_\alpha\wedge X_\beta\wedge X_\gamma}\omega),\delta) -
\Tilde{l_2^m}((0,-\iota_{X_\alpha\wedge X_\beta\wedge X_\delta}\omega),\gamma)+
\Tilde{l_2^m}((0,-\iota_{X_\alpha\wedge X_\gamma\wedge X_\delta}\omega),\beta)-
\Tilde{l_2^m}((0,-\iota_{X_\beta\wedge X_\gamma\wedge X_\delta}\omega),\alpha)$$
$$=(0,2d\omega(X_{\alpha},X_{\beta},X_{\gamma},X_{\delta}))=(0,0).$$
But despite this fact,
$\ \displaystyle \iota_{X_\alpha\wedge X_\beta\wedge X_\gamma}(\omega)$
belongs to $C^{\infty}(M)$ 
and not to $C^{\infty}_{Ham}(M)$, as can be seen from the definition of $\Tilde{l_3}$ since
$\displaystyle\ pr_1 \circ \tilde{l_3}(\alpha,\beta,\gamma)= -\iota_{X_\alpha\wedge X_\beta\wedge X_\gamma}(\omega).$
\end{enumerate}
\end{rema}
\medskip

In order to account for the ambiguity of the bivector associated to a Hamiltonian function, we refine the definition of $D^2(M,\omega)$.

\begin{defn}
Let  $(M,\omega)$ be a 2-plectic manifold. 
A pair $(f,v)$ in $C^{\infty}(M)\times \mathfrak{X}^2(M)$ is said to be \emph{an Hamiltonian pair} if
 $df=-\iota_{v}\omega$. We denote by $Ham^0(M)$ the set of Hamiltonian pairs.
\end{defn}

In the following proposition, we describe a new Lie 2-algebra of observables 
$\mathsf{\widetilde{D^2(M,\omega)}}$, obtained by including the Hamiltonian pairs. For simplicity, we use the same notations ($\Tilde{l_1}$, $\Tilde{l_2}$, $\Tilde{l_3}$) both for 
$\mathsf{D^2(M,\omega)}$ and $\mathsf{\widetilde{D^2(M,\omega)}}$.

\begin{prop}
Let $(M,\omega)$ be a 2-plectic manifold. Then the graded vector space
$$\mathsf{\widetilde{D^2(M,\omega)}}=(C^{\infty}(M)\times Ham^0(M))\oplus  \Omega^1_{Ham}(M),$$ where the elements of $C^{\infty}(M)\times Ham^0(M)$ are of degree -1 and the elements of $\Omega^1_{Ham}(M)$ are of degree 0, together with the following brackets :
\begin{itemize}
\item $\Tilde{l_1} : C^{\infty}(M)\times Ham^0(M)\mapsto \Omega_{Ham}^1(M)$,  
$\Tilde{l_1}(\Tilde{f},(f,v))=d\circ pr_1(\Tilde{f},(f,v))= d\Tilde{f}$, with $pr_1(\Tilde{f},(f,v))=\Tilde{f}$ 
\item $\Tilde{l_2}$  is given by two antisymmetric maps $\Tilde{l_2^p}$
and $\Tilde{l_2^m}$ \\
$\Tilde{l_2^p }: \bigwedge^2(\Omega^1_{Ham}(M))\longrightarrow  \Omega^1_{Ham}(M), \, \, 
\Tilde{l_2^p} (\a,\b)= \iota_{X_{\a}\wedge X_{\b}}(\omega)\, ,$ where
$d\iota_{X_{\a}\wedge X_{\b}}(\omega)=-\iota_{[X_\alpha,X_\beta]_S}(\omega)$
and
$\Tilde{l_2^m} : (C^{\infty}(M)\times Ham^0(M))\times \Omega^1_{Ham}(M) \longrightarrow C^{\infty}(M)\times Ham^0(M),$ 
$$\Tilde{l_2^m}((\Tilde{f}, (f,v)),\alpha)=(0,(\iota_{v\wedge X_\alpha}(\omega), [v,X_{\alpha}]_S)), $$
($\Tilde{l_2^m}$ is  antisymmetric upon setting
$\Tilde{l_2^m}(\a, (\Tilde{f},(f,v)))=-(0,(\iota_{v\wedge X_\alpha}(\omega),[v,X_{\alpha}]_S)))$
\item $\Tilde{l_3}: \bigwedge^3( \Omega^1_{Ham}(M)) \longrightarrow C^{\infty}(M)\times Ham^0(M)$, $\Tilde{l_3}(\alpha,\beta,\gamma)=
(-\iota_{X_\alpha\wedge X_\beta\wedge X_\gamma}(\omega), (0,0))$
\end{itemize} 
defines a Lie 2-algebra $(\mathsf{\widetilde{D^2(M,\omega)}}, \Tilde{l_1}, \Tilde{l_2}, \Tilde{l_3})$.
\end{prop}

\begin{proof} 
The proof is analogous to the proof of the preceding proposition.
\end{proof}

\begin{rema}

Going from $D^2(M,\omega)$ to $\mathsf{\widetilde{D^2(M,\omega)}}$ is similar to the construction of observables in the pre-n-plectic case, i.e., replacing $\Omega^{n-1}_{Ham}(M)$  by $Ham^{n-1}(M)=\{(\alpha, X)\in \Omega^{n-1}_{Ham}(M)\times \mathfrak{X}^1(M)\ \vert \iota_X\omega=d\alpha\}$ (as in \cite{CFRZ}).
In fact, the above Lie 2-algebra $\mathsf{\widetilde{D^2(M,\omega)}}$ could be replaced by $(C^{\infty}(M)\oplus Ham^0(M))\oplus Ham^1(M)$ to accommodate the pre-2-plectic case.

\end{rema}

\medskip

We relate the Lie 2-algebra $\mathsf{\widetilde{D^2(M,\omega)}}$ of Delgado's observables, via three easy but important propositions, to Rogers's observables $\mathsf{L^2(M,\omega)}$ and to the Lie 2-algebra of (truncated) multivector fields on $M$.

\begin{prop}\label{I}
 The couple $I=(I_1,I_2)$ with $I_{1,0}=\mathrm{id}$, $I_{1,-1}(\tilde{f}) = (\tilde{f},(0,0))$ for all $\tilde{f}\in C^{\infty}(M)$,  and $I_2=0$ is a Lie 2-morphism from $\mathsf{L^2(M,\omega)}$ to $\mathsf{\widetilde{D^2(M,\omega)}}$, which is injective in each degree. Thus $\mathsf{L^2(M,\omega)}$ is a sub Lie 2-algebra of $\mathsf{\widetilde{D^2(M,\omega)}}$.
\end{prop}
\begin{proof}
Straightforward verification. 
\end{proof}

\begin{prop}\label{S}
 The couple $\Phi=(\Phi_1,\Phi_2)$ with $\Phi_{1,0}=\mathrm{id}$, $\Phi_{1,-1}=pr_1$ and $\Phi_2=0$ is a Lie 2-morphism from $\mathsf{\widetilde{D^2(M,\omega)}}$ to $\mathsf{L^2(M,\omega)}$, which is surjective in each degree. Furthermore $\Phi \circ I = \mathrm{id}$.
\end{prop}
\begin{proof}
Straightforward verification. 
\end{proof}

\begin{defn} \label{gradient}
 A \emph{2-plectic gradient (map)} is a Lie 2-algebra morphism 
$$\Psi=(\Psi_1,\Psi_2) : \mathsf{\widetilde{D^2(M,\omega)}}=(C^{\infty}(M)\times Ham^0(M)) \oplus\Omega^1_{Ham}(M) \longrightarrow \mathfrak{X}^2(M)\oplus\mathfrak{X}^1(M)$$ satisfying
\begin{enumerate}[(1)]
    \item $\Psi_{1,-1}(\Tilde{f},(0,0))=0$ for all $\Tilde{f}\in C^\infty(M)$.
    \item $\Psi_{1,-1}(0,(f,v))=v$ for all $(f,v)$ in $Ham^0(M)$.
    \item $\Psi_{1,0}(\alpha)={X_\alpha}$ for all $\alpha\in\Omega^1_{Ham}(M)$ with $X_\alpha$ in $\mathfrak{X}^1(M)$ fulfilling $d\alpha=-\iota_{X_\alpha}\omega$.
\end{enumerate}
\end{defn}

\begin{prop} \label{psi-prop6}
Let  $\Psi_1 :\mathsf{\widetilde{D^2(M,\omega)}} \longrightarrow \mathfrak{X}^2(M)\oplus\mathfrak{X}^1(M)$ be a graded linear map and  
$\Psi_2 : \Lambda^2\Omega^1_{Ham}(M) \to \mathfrak{X}^2(M)$ be a linear map.
Then $\Psi=(\Psi_1,\Psi_2)$ defines a Lie 2-morphism 
if and only if the following conditions are satisfied for
 $ \Tilde{f}\in C^{\infty}(M)$, $(f,v)\in Ham^0(M)$ and $\alpha, \beta,\gamma \in \Omega^1_{Ham}(M)$ :
 \begin{itemize}
\item[$(A_1)$] for $\Psi$ : $\Psi_{1,0}\circ (d\circ pr_1(\Tilde{f}, (f,v)))=X_{d\Tilde{f}}=0$ 
\item[$(A_2)$] for $\Psi$ : $\Psi_{1,0}(i_{X_\alpha\wedge X_\beta}\omega)=[\Psi_{1,0}(\alpha), \Psi_{1,0}(\beta) ]_S$
\item[$(A_3)$] for $\Psi$ : $\Psi_{1,-1}(\Tilde{{l}^m_2}((\Tilde{f},(f,v)),\alpha)=[\Psi_{1,-1}(\Tilde{f},(f,v)), \Psi_{1,0}(\alpha)]_S$ 

\noindent and $\Psi_2(d\circ pr_1((\Tilde{f},(f,v)),\alpha)=\Psi_2(d\Tilde{f},\alpha)=0$
\item[$(A_4)$] for $\Psi$ : $\Psi_2(i_{X_\alpha\wedge X_\beta}\omega,\gamma)+c.p.=[X_\alpha, \Psi_2(\beta,\gamma)]_S+ c.p.$
\end{itemize}
Furthermore, $\Psi=(\Psi_1,\Psi_2)$ with the conditions (1)-(3) of the preceding definition  is a 2-plectic gradient if and only if $\Psi_2$
satisfies $\Psi_2(d\Tilde{f},\alpha)=0$ and $(A_4)$ (for $\Psi$). 
If $\Psi_2=0$, the conditions (1)-(3) of the preceding definition already imply $(A_1)-(A_4)$ (for $\Psi$). 
\end{prop}

\begin{proof} 
The relations $(A_1)-(A_4)$ for $\Psi$  are the specialisation of the relations $(A_1)-(A_4)$ of the definition of a Lie 2-algebra morphism to the case at hand. 
Obviously, $\Psi_2(d\Tilde{f},\alpha)=0$ follows from the vanishing of the other terms in the condition $(A_3)$ if $(f,v)=0$, i.e., the relation $(A_3)$ of a Lie 2-algebra morphism is equivalent for $\Psi$ to the relation $\Psi_2(d\Tilde{f},\alpha)=0$ and 
$$(A_3')\quad \, \Psi_{1,-1}(\Tilde{{l}^m_2}((\Tilde{f},(f,v)),\alpha)=[\Psi_{1,-1}(\Tilde{f},(f,v), \Psi_{1,0}(\alpha)]_S \, . \hfill$$
Furthermore, for 
any graded linear map, the "rules" (1)-(3) immediately imply $(A_1)$ (for $\Psi$) and $(A_2)$
(for $\Psi$),
and also the fact that the LHS of $(A_3')$ is equal to
$\Psi_{1,-1}((0, (\iota_{v\wedge X_{\alpha}}(\o),[v,X_{\alpha}]_{S} ))=[v, X_{\alpha}]_S , $
which is precisely also the RHS of $(A_3')$.
It is thus equivalent to say that a graded linear map
$\Psi=(\Psi_1,\Psi_2)$ with (1)-(3) is a 2-plectic gradient and to say that $\Psi_2$ satisfy 
the conditions $\Psi_2(d\Tilde{f},\alpha)=0$ and $(A_4)$. These two last conditions are trivially guaranteed when $\Psi_2=0$.
\end{proof}

\medskip
\begin{rema}\label{psi-rem6}\hfill

\begin{enumerate}
   \item 
   We always obtain a (unique!) strict 2-plectic gradient $\Psi$ by setting $\Psi_2=0$ and  $\Psi_1$ as stipulated by $(1)-(3)$ of Definition \ref{gradient}.
\item 
Since $\Psi_2$ can be chosen to be non zero, we provide definitions and formulas valid in the general case. Nevertheless, we mainly consider the case where $\Psi$ is the strict 2-plectic gradient described in the first part of this remark. 
\end{enumerate}
\end{rema}

\section{Comomentum maps}

In this section, we suppose that $(M,\omega)$ is a 2-plectic manifold, and we introduce the notion of 2-comomentum maps, which are $\mathsf{\widetilde{D^2(M,\omega)}}$-valued Lie 2-morphisms compatible with a given 2-action $\rho$  of a general Lie 2-algebra $\mathfrak{g}$ on $M$. 
We will also explain how the notion of ``homotopy (co-)moment maps'' (see \cite{CFRZ, MaZa, RyWu1, Mi}) fits into our more general framework.

\subsection{The notion of 2-comomentum maps}  \hfill

\bigskip
Let $(M,\omega)$ be a 2-plectic manifold and $\rho: \mathfrak{g}=\mathfrak{g}_{-1}\oplus\mathfrak{g}_{0} \longrightarrow \mathfrak{X}^2(M)\oplus\mathfrak{X}^1(M)$ a 2-action of a Lie 2-algebra  $(\mathfrak{g}_{-1}\oplus\mathfrak{g}_{0}, l_1, l_2, l_3)$ on $M$. 
We give a notion of 2-comomentum map for $(M,\omega)$ compatible with $\rho$.
We replace the usual target of the comomentum map $\mathsf{L^2(M,\omega)}$ (compare Section 7.2) by $\mathsf{\widetilde{D^2(M,\omega)}}$ since it allows $\mathfrak{g}_{-1}$ to act (by the very definition of a 2-action).

\subsubsection{General definitions}

\begin{defn}
Let $\Psi : \mathsf{\widetilde{D^2(M,\omega)}} =\big(C^{\infty}(M)\times Ham^0(M)\big)\oplus\Omega^1_{Ham}(M) \longrightarrow \mathfrak{X}^2(M)\oplus\mathfrak{X}^1(M)$ be a 2-plectic gradient, i.e., a Lie 2-morphism such that $\Psi_{1,0}(\alpha)=X_\alpha$, $\Psi_{1,-1}((\tilde{f},(f,v)))=v$  
 where $(f,v)$ fulfills $\iota_{v}\omega=-df$ and $X_\alpha\in \mathfrak{X}^1(M)$ is given by $\iota_{X_\alpha}\omega=-d\alpha$.

\begin{enumerate}
    \item A \emph{(Delgado)} {\emph{2-comomentum map}} {\emph{or comomentum map}}, $\lambda$, \emph{for the 2-action $\rho$ along $\Psi$} is a lift of this action to the following commutative diagram of Lie 2-morphisms :

\[
\begin{tikzpicture}[>=latex, node distance=2.5cm, auto]
  \node (A) at (0,0) {$\mathfrak{g} = \mathfrak{g}_{-1}\oplus\mathfrak{g}_{0}$};
  \node (B) at (4,4) {$\mathsf{\widetilde{D^2(M,\omega)}}=\big(C^{\infty}(M)\times Ham^0(M)\big)\oplus\Omega^1_{Ham}(M)$};
  \node (E) at (4,-1) {$\mathfrak{X}^2(M)\oplus\mathfrak{X}^1(M)$};

  \draw[->] (A) -- node[above left] {$\lambda$} (B);
  \draw[->] (B) -- node[above right] {$\Psi$}(E);
  \draw[->] (A) -- node[above] {$\rho$}(E);
\end{tikzpicture}
\]

\bigskip

\item An action $\rho$ is called \emph{(strongly) Hamiltonian along $\Psi$ (or an action with comomentum map}), if there exists a comomentum map for $\rho$ along $\Psi$.
\end{enumerate}
\end{defn}

    \begin{rema} \label{JLie2}
Since $\lambda$ is Lie 2-morphism, it can be decomposed into components ${\lambda}_{1,0}:\mathfrak{g}_0 \longrightarrow \Omega^1_{Ham}(M)$, \
     ${\lambda}_{1,-1}:\mathfrak{g}_{-1} \longrightarrow C^{\infty}(M)\times Ham^0(M)$ and ${\lambda}_{2}:\bigwedge^2\mathfrak{g}_0 \longrightarrow C^{\infty}(M)\times Ham^0(M)$ satisfying the following relations $\forall a\in \mathfrak{g}_{-1}$ and $\forall x, y, z\in \mathfrak{g}_0$:
   
   \begin{itemize}
    \item[$(A_1)$] for ${\lambda}$: ${\lambda}_{1,0}(l_1(a))=d\circ pr_1({\lambda}_{1,-1}(a))$
     \item[$(A_2)$] for ${\lambda}$: ${\lambda}_{1,0}(l_2^p(x,y))=\iota_{X_{{\lambda}_{1,0}(x)}\wedge X_{{\lambda}_{1,0}(y)}}\omega+ d\circ pr_1(\lambda_2(x,y))$
     \item[$(A_3)$] for ${\lambda}$: ${\lambda}_{1,-1}(l_2^m(a,x))=\Tilde{{l}^m_2}({\lambda}_{1,-1}(a),{\lambda}_{1,0}(x)) + {\lambda}_2(l_1(a),x)$
     \item[$(A_4)$] for ${\lambda}$: $\big{(}{\lambda}_2(l_2^p(x,y),z)+c.p.\big{)}+ {\lambda}_{1,-1}(l_3(x,y,z))=(-\iota_{X_{{\lambda}_{1,0}(x)}\wedge X_{{\lambda}_{1,0}(y)}\wedge X_{{\lambda}_{1,0}(z)}}\omega\ ,\ 0)+ \big{(}\Tilde{{l}^m_2}({\lambda}_{1,0}(x),{\lambda}_2(y,z))+c.p.\big{)}$,
    \end{itemize}
    where $pr_1$  denotes the projection on the first component of $C^{\infty}(M)\times Ham^0(M)$.
    \end{rema}
    
\begin{prop}
    \label{Jcomposition}
The equality $\Psi\circ {\lambda}=\rho$ implies $\forall a\in \mathfrak{g_{-1}}$ and $\forall x,y\in \mathfrak{g_0}$ : 
    \begin{itemize}
        \item[$(C_1)$] $\Psi_{1,0}({\lambda}_{1,0}(x))=X_{{\lambda}_{1,0}(x)}=\rho_{1,0}(x)$, i.e., $\iota_{\rho_{1,0}(x)}\omega =- d{\lambda}_{1,0}(x)$
        \item[$(C_2)$] $\Psi_{1,-1}({\lambda}_{1,-1}(a))=\rho_{1,-1}(a)$, i.e., $\iota_{\rho_{1,-1}(a)}\omega = -d\circ pr_{2,1} ({\lambda}_{1,-1}(a))$
        \item[$(C_3)$] $\Psi_{1,-1}({\lambda}_2(x,y))+\Psi_2({\lambda}_{1,0}(x),{\lambda}_{1,0}(y))=\rho_2(x,y)$, i.e.,
        $$\iota_{\big{(}\rho_{2}(x,y)-\Psi_2({\lambda}_{1,0}(x), {\lambda}_{1,0}(y)\big{)}}\omega =-d\circ  pr_{2,1}({\lambda}_2(x,y)).$$
    \end{itemize}
     where $pr_{2,1}$  is defined on $C^{\infty}(M)\times Ham^0(M)$ by
     $${pr_{2,1}}(\Tilde{f},(f,v))=pr_1\circ pr_2(\Tilde{f},(f,v)) =f.$$
\end{prop}

\begin{proof} 
The proof is obtained upon specialising Lemma \ref{compositionmorphismes}.\end{proof}

\begin{cor}
    If ${\lambda}$ is a 2-comomentum map for $\rho$ along  \emph{any} $\Psi$ then $\forall a\in \mathfrak{g}_{-1}$ and $\forall x, y\in \mathfrak{g}_0$,
     $\iota_{\rho_{1,0}(x)}\omega$, $\iota_{\rho_{1,-1}(a)}\omega$ and $\iota_{\big{(}\rho_{2}(x,y)-\Psi_2({\lambda}_{1,0}(x), {\lambda}_{1,0}(y)\big{)}}\omega$ are exact. Thus a Hamiltonian 2-action is quasi-2-plectic.
\end{cor}

\begin{rema}
        The equations $(A_i)$ and $(C_j)$ ($i=1,\ldots, 4$; $j=1,\ldots, 3$) will be key tools to study examples (see Appendix \ref{app-ex}).
\end{rema}

If we do not assume that ${\lambda}$ is a Lie 2-morphism, we obtain a weakened notion of a comomentum map. 

\begin{defn}
Let $(M,\omega)$ be a 2-plectic manifold. A Lie 2-action $\rho : \mathfrak{g}_{-1}\oplus\mathfrak{g}_{0} \longrightarrow \mathfrak{X}^2(M)\oplus\mathfrak{X}^1(M)$
is called \emph{weakly Hamiltonian along $\Psi$} if there exists a map ${\lambda}:\mathfrak{g}_{-1}\oplus\mathfrak{g}_{0} \longrightarrow \mathsf{\widetilde{D^2(M,\omega)}}$, that can be decomposed into linear components ${\lambda}_{1,0}:\mathfrak{g}_0 \longrightarrow \Omega^1_{Ham}(M)$, \
     ${\lambda}_{1,-1}:\mathfrak{g}_{-1} \longrightarrow C^{\infty}(M)\times Ham^0(M)$ and ${\lambda}_{2}:\bigwedge^2\mathfrak{g}_0 \longrightarrow C^{\infty}(M)\times Ham^0(M)$, lifting the action along $\Psi$.
\end{defn} 

\begin{prop}
   A 2-plectic action is weakly Hamiltonian along  $\Psi$ if and only if there exist linear maps ${\lambda}_{1,0}, {\lambda}_{1,-1}$ and ${\lambda}_{2}$ such that $(C_1), (C_2)$ and $(C_3)$ hold.
\end{prop}

\begin{proof}
    Since $\Psi\circ {\lambda}=\rho$, the result follows from Proposition \ref{Jcomposition}.
    
\end{proof}

\begin{prop}
    A strongly Hamiltonian 2-action is weakly Hamiltonian.
\end{prop}

\begin{proof}
    Obvious.
\end{proof}

\subsubsection{The fundamental case}
 
\begin{defn}\hfill
\begin{enumerate}
    \item  A 2-comomentum map for $\rho$ along $\Psi$ is called \emph{fundamental} if $\Psi$ is a strict Lie 2-morphism (that is, $\Psi_2 = 0$) (see Proposition \ref{psi-prop6}).
    \item A strongly (or weakly) Hamiltonian 2-action along $\Psi$ is called \emph{fundamental} if $\Psi$ is strict.
\end{enumerate}
   
\end{defn}

\begin{prop}
    If ${\lambda}$ is a  fundamental 2-comomentum map, then
\begin{enumerate}
    \item $\iota_{\rho_{1,0}(x)}\omega = -d{\lambda}_{1,0}(x)$,
    \item $\iota_{\rho_{1,-1}(a)}\omega = -d\circ pr_{2,1} ({\lambda}_{1,-1}(a))$ 
    \item $\iota_{\rho_{2}(x,y)}\omega =- d\circ pr_{2,1} ({\lambda}_{2}(x,y))$,
    
\end{enumerate}
$\forall a\in \mathfrak{g}_{-1}$ and $\forall x, y\in \mathfrak{g}_0$.
\end{prop}
    
\begin{proof} 
   Use Proposition \ref{Jcomposition} with $\Psi_2=0$.
\end{proof}

\begin{prop}
    If ${\lambda}$ is a fundamental 2-comomentum map, then
   $\iota_{\rho_{1,0}(x)}\omega$, $\iota_{\rho_{1,-1}(a)}\omega$ and $\iota_{\rho_{2}(x,y)}\omega$ are exact. Thus a (weakly or strongly) fundamental Hamiltonian 2-action is 2-plectic.
\end{prop}
\begin{proof}
    They are exact because of the previous proposition, so they are closed. And, since $\omega$ is closed 
    \begin{itemize}
        \item $\mathcal{L}_{\rho_{1,0}(x)}\omega=0$ $\Longleftrightarrow$ $d\circ \iota_{\rho_{1,0}(x)}\omega=0$, 
\item $\mathcal{L}_{\rho_{1,-1}(a)}\omega=0,$ $\Longleftrightarrow$ $d\circ \iota_{\rho_{1,-1}(a)}\omega=0$, 
\item $\mathcal{L}_{\rho_{2}(x,y)}\omega=0 \,$$\Longleftrightarrow$ $d\circ \iota_{\rho_{2}(x,y)}\omega=0$.
 \end{itemize}
\end{proof}

\begin{rema}
    In the examples (see Appendix \ref{app-ex}) we will only consider  fundamental 2-comomentum maps.
\end{rema}

\subsection{2-comomentum maps and homotopy moment maps}\hfill

\bigskip

In this subsection, we will show that our definition generalises the notion of homotopy moment maps, associated to a Lie algebra action $\rho^{\mathfrak{g}_0}$ of a Lie algebra $\mathfrak{g}_0$ on a 2-plectic manifold (see \cite{CFRZ} and \cite{MaZa}).

\subsubsection{Recap' of homotopy moment maps}

\begin{defn}
    [\cite{CFRZ}]
    Let $(M,\omega)$ be a $n$-plectic manifold, let $\mathfrak{g}_0$ be a Lie algebra and let $\rho^{\mathfrak{g}_0}$ be an infinitesimal action of $\mathfrak{g}_0$ on $M$ which associates to $x$ in $\mathfrak{g}_0$ a Hamiltonian vector field $v_x$. A homotopy comomentum map (called ``homotopy moment map'' in \cite{CFRZ}) for this action is a $L_{\infty}$-morphism ${\lambda}^{\mathfrak{g}_0}$ from $\mathfrak{g}_0$ to $\mathsf{L^n(M,\omega)}$ 
     such that, for all $x$ in $\mathfrak{g}_0$, $-\iota_{v_x}\omega=d({\lambda}^{\mathfrak{g}_0}_1(x))$,  where ${\lambda}^{\mathfrak{g}_0}_1:\mathfrak{g}_0\rightarrow\Omega^{n-1}_{Ham}(M)$.
\end{defn}
  \noindent Let us specialise the previous definition to the case where $(M, \omega)$ is a 2-plectic manifold. Then $\mathsf{L^n(M,\omega)}$ is the Lie 2-algebra $\mathsf{L^2(M,\omega)}= C^{\infty}(M)\oplus\Omega^1_{Ham}(M)$ and we have the following important remark.

\begin{rema}
Let $(M,\omega)$ be a 2-plectic manifold, $\mathfrak{g}_0$ be a Lie algebra and let $\rho^{\mathfrak{g}_0}$ be the infinitesimal action of $\mathfrak{g}_0$ on $M$ which associates to $x$ the Hamiltonian vector field $v_x$. 
This action corresponds to a Lie 2-morphism  from $\{0\}\bigoplus \mathfrak{g}_0$ to 
$\{0\}\bigoplus \mathfrak{X}^1(M)$ (considered as Lie 2-algebras). Moreover, the 
homotopy comomentum map ${\lambda}^{\mathfrak{g}_0}$ corresponds to a Lie 2-morphism lift of the action of $\mathfrak{g}_0$ on $(M,\omega)$ along  ${(\Psi\circ I)}_{1,0}$, seen as a Lie 2-morphism that associates to $\alpha$ in $\Omega^1_{Ham}(M)$ the Hamiltonian vector field $X_{\alpha}$ and which is trivial in all other degrees, via the commutation of the diagram below.

\[
\begin{tikzpicture}[>=latex, node distance=2.5cm, auto]
  \node (A) at (0,0) {$\mathfrak{g}_{0}$};
  \node (B) at (3,3) {$\mathsf{L^2(M,\omega)}=C^{\infty}(M)\oplus\Omega^1_{Ham}(M)$};
  \node (C) at (3,0) {$\mathfrak{X}^1(M)$};
  \draw[->] (A) -- node[above left] {${\lambda}^{\mathfrak{g}_0}$} (B);
  \draw[->] (B) -- node[above right] {$(\Psi\circ I)_{1,0}$}(C);
  \draw[->] (A) -- node[above] {$\rho^{\mathfrak{g}_0}$}(C);
\end{tikzpicture}
\]

In the sequel, for simplicity, we will denote $(\Psi\circ I)_{1,0}$ by $\Psi_{1,0}$.
\end{rema}
\medskip
This notion has been generalised to Lie 2-algebras 
$\mathfrak{g}=\mathfrak{g}_{-1}\oplus \mathfrak{g}_0$, such that ($\mathfrak{g}_0, l_2)$ is a Lie algebra
(let us emphasize that this does not hold for a general Lie 2-algebra).

\begin{defn}[\cite{Mam}, \cite{MaZa}]\hfill 

Assume that $(M,\omega)$ is a 2-plectic manifold and that ($\mathfrak{g}=\mathfrak{g}_{-1}\oplus \mathfrak{g}_0, l_1,l_2,l_3)$ is a Lie 2-algebra such that
$(\mathfrak{g}_0,l_2)$ is a Lie algebra. Let $\rho^{\mathfrak{g}_0}$ be the infinitesimal action 
of $\mathfrak{g}_0$ on $M$, which associates to $x$ in $\mathfrak{g}_0$, the 
Hamiltonian vector field $v_x$.
    A homotopy comomentum map for the Lie 2-algebra $\mathfrak{g}$ is a Lie 2-morphism $${\lambda}^\mathfrak{g}:(\mathfrak{g}=\mathfrak{g}_{-1}\oplus\mathfrak{g}_0,l_1,l_2,l_3) \rightarrow (C^{\infty}(M)\oplus\Omega^1_{Ham}(M),\Tilde{l_1^R},\Tilde{l_2^R},\Tilde{l_3^R})$$ such 
    that $$-\iota_{v_x}\omega=d({\lambda}^\mathfrak{g}_{1,0}(x))$$
    for all $x$ in $\mathfrak{g}_0$, where ${\lambda}^\mathfrak{g}_{1,0}:\mathfrak{g}_0\rightarrow\Omega^{1}_{Ham}(M)$, ${\lambda}^\mathfrak{g}_{1,-1}:\mathfrak{g}_{-1}\rightarrow C^\infty(M)$ and ${\lambda}^\mathfrak{g}_2:\bigwedge^2(\mathfrak{g}_0)\rightarrow C^\infty(M)$.
\end{defn} 
The following remark enables us to rewrite this definition as a lift of  $\rho^{\mathfrak{g}_0}$
via a commutative diagram of Lie 2-morphisms. 
\begin{rema}
 Let $(M,\omega)$ be a 2-plectic manifold, $\mathfrak{g}=\mathfrak{g}_{-1}\bigoplus\mathfrak{g}_{0}$ such that $\mathfrak{g}_0$ is a Lie algebra and let $\rho^{\mathfrak{g}_0}$ be an infinitesimal action of $\mathfrak{g}_0$ on $M$ which associates to $x$ in $\mathfrak{g}_0$ a Hamiltonian vector field $v_x$. 
 As already observed in (\cite{Mam}, \cite{MaZa}), since ${\lambda}^\mathfrak{g}$ is a Lie 2-morphism, for all $a$ in $\mathfrak{g}_{-1}$, ${\lambda}^\mathfrak{g}_{1,0}(l_1(a))$ is exact and from the relation
 $$-\iota_{v_{l_1(a)}}\omega=d({\lambda}^\mathfrak{g}_{1,0}(l_1(a)))$$
we deduce that for all $a$ in $\mathfrak{g}_{-1}$, $\rho_{1,0}(l_1(a))=0$. That means that the 
action $\rho^{\mathfrak{g}_0}$ corresponds to a Lie 2-morphism  from $\mathfrak{g}_{-1}\bigoplus \mathfrak{g}_0$ to 
$\{0\}\bigoplus \mathfrak{X}^1(M)$, by putting $\rho_{1,0}=\rho^{\mathfrak{g}_0}$, 
 $\rho_{1,-1}=0$ and $\rho_{2}=0$ . Moreover, the 
homotopy comomentum map ${\lambda}^{\mathfrak{g}}$ corresponds to a lift of this action  on $(M,\omega)$ along  $\Psi_{1,0}$, seen as a Lie 2-morphism which associates to $\alpha$ in $\Omega^1_{Ham}(M)$ the Hamiltonian vector field $X_{\alpha}$, and which is trivial in all other degrees, via the commutation of the diagram below.   
\end{rema}

\[
\begin{tikzpicture}[>=latex, node distance=2.5cm, auto]
  \node (A) at (0,0) {$\mathfrak{g} = \mathfrak{g}_{-1}\oplus\mathfrak{g}_{0}$};
  \node (B) at (3,3) {$\mathsf{L^2(M,\omega)}=C^{\infty}(M)\oplus\Omega^1_{Ham}(M)$};
  \node (C) at (3,0) {$\mathfrak{X}^1(M)$};
  \draw[->] (A) -- node[above left] {${\lambda}^\mathfrak{g}$} (B);
  \draw[->] (B) -- node[above right] {$\Psi_{1,0}$}(C);
  \draw[->] (A) -- node[above] {$\rho^{\mathfrak{g}_0}$}(C);
\end{tikzpicture}
\]

For the sake of clarity, let us write all the relations that ${\lambda}^\mathfrak{g}$, $\rho^{\mathfrak{g}_0}=\rho_{1,0}$ and $\Psi_{1,0}$ satisfy.
\begin{prope}
Since ${\lambda}^\mathfrak{g}$ is Lie 2-morphism, we have $\forall a\in \mathfrak{g}_{-1}$, $\forall x, y, z\in \mathfrak{g}_0$ 
\begin{itemize}
\item[$(A_1)$] for ${\lambda}^\mathfrak{g}$ : ${\lambda}^\mathfrak{g}_{1,0}(l_1(a))=d({\lambda}^\mathfrak{g}_{1,-1}(a))$
\item[$(A_2)$] for ${\lambda}^\mathfrak{g}$ : ${\lambda}^\mathfrak{g}_{1,0}(l_2^p(x,y))=\iota_{((X_{{\lambda}^\mathfrak{g}_{1,0}(x)})\wedge (X_{{\lambda}^\mathfrak{g}_{1,0}(y)}))}\omega+ d({\lambda}^\mathfrak{g}_2(x,y))$
\item[$(A_3)$] for ${\lambda}^\mathfrak{g}$ : ${\lambda}^\mathfrak{g}_{1,-1}(l_2^m(a,x))= {\lambda}^\mathfrak{g}_2(l_1(a),x)$
\item[$(A_4)$] for ${\lambda}^\mathfrak{g}$ : $\big{(}{\lambda}^\mathfrak{g}_2(l_2^p(x,y),z)+c.p.\big{)}+ {\lambda}^\mathfrak{g}_{1,-1}(l_3(x,y,z))=-\omega(X_{{\lambda}^\mathfrak{g}_{1,0}(x)},X_{{\lambda}^\mathfrak{g}_{1,0}(y)}, X_{{\lambda}^\mathfrak{g}_{1,0}(z)})$.
\end{itemize}

The map $\Psi_{1,0}$ satisfies    
    
\begin{itemize}
    \item[$(A_1)$] for $\Psi_{1,0}$ : $\Psi_{1,0}\circ d=0$
\item[$(A_2)$] for $\Psi_{1,0}$ : $\Psi_{1,0}(\iota_{X_\alpha \wedge X_\beta}\omega)=[\Psi_{1,0}(\alpha), \Psi_{1,0}(\beta) ]_S$, $\forall \alpha,\beta\in\Omega^1_{Ham}(M)$.
     \end{itemize}

Since $\Psi \circ {\lambda}^\mathfrak{g}=\rho^{\mathfrak{g}_0}$, we have, $\forall x\in\mathfrak{g}_0$,
$X_{{\lambda}^\mathfrak{g}_{1,0}(x)}=\rho_{1,0}(x).$

The map $\rho^{\mathfrak{g}_0}=\rho_{1,0}$ satisfies

\begin{itemize}
    \item[$(A_1)$] for $\rho_{1,0}$ : $\rho_{1,0}\circ l_1(a)=0$, $\forall a \in \mathfrak{g}_{-1}$
\item[$(A_2)$] for $\rho_{1,0}$ : $\rho_{1,0}(l_2(x,y))=[\rho_{1,0}(x), \rho_{1,0}(y) ]_S$, 
$\forall x, y \in \mathfrak{g}_0$.
     \end{itemize}

\end{prope}

The fact that in this context the action reduces to $\rho_{1,0}$ 
($\rho_{1,-1}=0$ and $\rho_2=0$ since the target of the action is $\mathfrak{X}^1(M)$ and not
$\mathfrak{X}^2(M)\oplus\mathfrak{X}^1(M)$ contrary to our 2-actions) and that
$\rho_{1,0}$ has to satisfy
$\rho_{1,0}(l_1(a))=0$, for any $a$ in $\mathfrak{g}_{-1}$,
explains why  the  authors of \cite{Mam} and \cite{MaZa}
are led to assume (without any lost of generality) that 
their Lie 2-algebra $\mathfrak{g}=\mathfrak{g}_{-1}\bigoplus \mathfrak{g}_{0}$  (where $\mathfrak{g}_0$ is a Lie algebra) is skeletal.

\begin{rema}
     Even for $\rho^{\mathfrak{g}_0}=0$ a homotopy comomentum map can be non trivial. In fact, if $\rho^{\mathfrak{g}_0}=0$, then ${\lambda}^\mathfrak{g}_{1,0}(x)$ is closed $\forall x\in \mathfrak{g}_0$, since $d({\lambda}^\mathfrak{g}_{1,0}(x))=-\iota_{\rho_{1,0}(x)}\omega$  but neither ${\lambda}^\mathfrak{g}_{1,0}$, ${\lambda}^\mathfrak{g}_{1,-1}$ nor ${\lambda}^\mathfrak{g}_2$ have to be zero.
    \end{rema}
    \begin{example}

    Take the Lie 2-algebra $(\mathfrak{g}=\mathfrak{g}_{-1}\oplus \mathfrak{g}_0,\ l_1,\ l_2 = 0,\ l_3 = 0)$ with $\mathfrak{g}_0$ an abelian Lie algebra which can be written  as $\mathfrak{g}_0=\mathfrak{g}_{0,0}\oplus \mathrm{im}(l_1)$, compare Example 2 of Appendix C. 

    We define ${\lambda}^\mathfrak{g}$ degree-wise:
\begin{itemize}
\item ${\lambda}^\mathfrak{g}_{1,0}:\mathfrak{g}_{0,0}\oplus \mathrm{im}(l_1)\longrightarrow \Omega^1_{Ham}(M)$ linear with ${{\lambda}^\mathfrak{g}_{1,0}}_{\vert \mathrm{im}(l_1)}=0$ and ${{\lambda}^\mathfrak{g}_{1,0}}(x)=d(f_x)$, where $f_x\in C^{\infty}(M)$, $\forall x\in \mathfrak{g}_{0,0}$,
\item ${\lambda}^\mathfrak{g}_{1,-1}:\mathfrak{g}_{1,-1} \longrightarrow C^{\infty}(M)$ with ${\lambda}^\mathfrak{g}_{1,-1}\in \mathfrak{g}_{1,-1}^*$,
\item ${\lambda}^\mathfrak{g}_{2}:\bigwedge^2(\mathfrak{g}_{0,0}\oplus \mathrm{im}(l_1))\longrightarrow C^{\infty}(M)$ with ${{\lambda}^\mathfrak{g}_{2}}_{\vert \bigwedge^2(\mathfrak{g}_{0,0})}\in \bigwedge^2\mathfrak{g}_{0,0}^*$ and ${\lambda}^\mathfrak{g}_2(x,y)=0$, whenever $x$ or $y$ is in $\mathrm{im}(l_1)$.
     \end{itemize}
    \end{example}
    

\subsubsection{Homotopy moment maps as 2-comomentum maps}\hfill

In the previous subsection, we have seen that the  notion of homotopy moment map in \cite{CFRZ, MaZa} and thus also in \cite{RyWu1, Mi} can be seen as a lift of an infinitesimal action $\rho^{\mathfrak{g}_0}$ of a Lie algebra $\mathfrak{g}_0$  via a commutative diagram of Lie 2-morphisms.
It remains here to see that such homotopy moment maps are particular cases of 2-momentum maps.
Let us first illustrate this assertion by the following diagram.

\[
\begin{tikzpicture}[>=latex, node distance=2.5cm, auto]
  \node (A) at (0,0) {$\mathfrak{g} = \mathfrak{g}_{-1}\oplus\mathfrak{g}_{0}$};
  \node (B) at (6,6) {$\mathsf{\widetilde{D^2(M,\omega)}}=\big(C^{\infty}(M)\times Ham^0(M)\big)\oplus\Omega^1_{Ham}(M)$};
  \node (C) at (4,2) {$\mathsf{L^2(M,\omega)}$};
  \node (D) at (4,0) {$\mathfrak{X}^1(M)$};
  \node (E) at (6,-2) {$\mathfrak{X}^2(M)\oplus\mathfrak{X}^1(M)$};
   \node (F) at (5.7,5.8) { };
   \node (G) at (4.3,2.2) { };
    \node (H) at (5.5,5.8) { };
  
  \draw[->] (A) -- node[above left] {${\lambda}$} (H);
  \draw[->] (A) -- node[above] {${\lambda}^\mathfrak{g}$}(C);
  \draw[->] (B) -- node[above right] {$\Psi$}(E);
  \draw[->] (F) -- (C);
  \draw[->] (G) -- (B);
  \draw[->] (C) -- node[above right] {$\Psi_{1,0}$}(D);
  \draw[->] (E) -- (D);
  \draw[->] (A) -- node[above] {$\rho_{1,0}$}(D);
  \draw[->] (A) -- node[above] {$\rho$}(E);
  \draw[->] (E) -- (D);
   \draw[->] (F) -- node[above] {$\Phi\ $}(C);
    \draw[->] (G) -- node[above right] {$\ I\ $}(B);
\end{tikzpicture}
\]

Let us recall that, in this diagram, $M$ is a 2-plectic manifold, 
($\mathfrak{g}_{-1}\oplus \mathfrak{g}_0, l_1,l_2,l_3)$ is a Lie 2-algebra such that
$(\mathfrak{g}_0,l_2)$ is a Lie algebra, $I$ is the Lie 2-morphism injecting $\mathsf{L^2(M,\omega)}$ in $\mathsf{\widetilde{D^2(M,\omega)}}$, constructed in proposition \ref{I}, $\Phi$ is the Lie 2-morphism constructed in proposition \ref{S}, $\Psi$ is the Lie 2-morphism 
from  $\mathsf{\widetilde{D^2(M,\omega)}}=(C^{\infty}(M)\times Ham^0(M)) \oplus\Omega^1_{Ham}(M)$ into $\mathfrak{X}^2(M)\oplus\mathfrak{X}^1(M)$, studied and employed before (here $\Psi$ is supposed to be strict).

The precise correspondence between homotopy comomentum maps and 2-comomentum maps can now easily be described. First, if ${\lambda}^\mathfrak{g}$ is a homotopy comomentum map for the Lie 2-algebra $\mathfrak{g}$ in the sense of \cite{MaZa}, i.e., $\displaystyle {\lambda}^\mathfrak{g}:(\mathfrak{g}_{-1}\oplus\mathfrak{g}_0,l_1,l_2,l_3) \rightarrow (\mathsf{L^2(M,\omega)},\Tilde{l_1^R},\Tilde{l_2^R},\Tilde{l_3^R})$ such that 
    $-\iota_{v_x}\omega=d({\lambda}^\mathfrak{g}_{1,0}(x))$, $\forall x\in\mathfrak{g}_0$, with ${\lambda}^\mathfrak{g}_{1,0}:\mathfrak{g}_0\rightarrow\Omega^{1}_{Ham}(M)$, ${\lambda}^\mathfrak{g}_{1,-1}:\mathfrak{g}_{-1}\rightarrow C^\infty(M)$ and ${\lambda}^\mathfrak{g}_2:\bigwedge^2(\mathfrak{g}_0)\rightarrow C^\infty(M)$), then 
    $\rho: \mathfrak{g}_{-1}\oplus \mathfrak{g}_0\rightarrow \mathfrak{X}^2(M)\oplus
    \mathfrak{X}^1(M)$ defined by $\rho_{1,0}=\rho^{\mathfrak{g}_0}$, $\rho_{1,-1}=0$, $\rho_2=0$, is a 2-action. Moreover, the map  ${\lambda}=I\circ {\lambda}^\mathfrak{g} : \mathfrak{g}_{-1}\oplus\mathfrak{g}_0 \rightarrow \mathsf{\widetilde{D^2(M,\omega)}}$ is a 2-comomentum map for this 2-action. 

    \smallskip
    
\noindent Conversely, let 
   $\rho:\mathfrak{g}_{-1}\oplus \mathfrak{g}_0 \rightarrow\mathfrak{X}^2(M)\oplus\mathfrak{X}^1(M)$ be a 2-action with $\rho_{1,-1}=0$ and $\rho_{2}=0$ and let 
   $${\lambda} : (\mathfrak{g}_{-1}\oplus\mathfrak{g}_0,l_1,l_2,l_3) \rightarrow (\mathsf{\widetilde{D^2(M,\omega)}},\Tilde{l_1},\Tilde{l_2},\Tilde{l_3})$$ be a 2-comomentum map for this 2-action. Then 
   ${\lambda}^\mathfrak{g}=\Phi\circ {\lambda} : (\mathfrak{g}_{-1}\oplus\mathfrak{g}_0,l_1,l_2,l_3) \rightarrow \mathsf{L^2(M,\omega)}$
   is a homotopy comomentum map for the Lie 2-algebra $\mathfrak{g}_{-1}\oplus \mathfrak{g}_0$ in the sense of
    \cite{MaZa} and $\rho_{1,0}$ defines  an infinitesimal action $\rho^{\mathfrak{g}_0}=\rho_{1,0}$ 
    of the Lie algebra $\mathfrak{g}_0$ on $M$ by Hamiltonian vector fields.
    



\appendix

\section{Quasi-isomorphism with skeletal Lie 2-algebra}

In this appendix, we will give a constructive proof of Proposition \ref{quasi-iso}. For the convenience of the reader we did not put this proof in the main body of the paper. 
We also give examples of concrete applications of this proposition.

\subsection{Proof}\hfill

\medskip

\noindent {\bf Claim :} Every Lie 2-algebra $(L_{-1}\oplus L_0,l_1,l_2,l_3)$ is quasi-isomorphic to a skeletal Lie 2-algebra $(\overline{L_{-1}}\oplus\overline{L_0}, \overline{l_1}=0, \overline{l_2}, \overline{l_3})$ with $\overline{L_{-1}}= \ker(l_1)$ and $\overline{L_0}=L_0/
\mathrm{im}(l_1)$.
\vskip0,5cm
\begin{proof}
We choose decompositions of the vector spaces $L_{-1}$ and $L_0$ as follows
$$L_{-1}= \ker(l_1) \oplus C \quad {\hbox{and}}\quad L_{0} = \mathrm{im}(l_1)\oplus C'$$
where the sub vector space $C$ is a complement to $\ker(l_1)$ and $C'$ is 
a complement to $\mathrm{im}(l_1)$. We have to construct to define the mappings $\overline{l_2}$, that is $\overline{l_2^p}$ and $\overline{l_2^m}$, and $\overline{l_3}$,
as well as a Lie 2-algebra morphism from $L_{-1}\oplus L_0$
to $\overline{L_{-1}}\oplus\overline{L_0}$. 
First, we define the maps $F_{1,0} = L_0\rightarrow \overline{L_0}$, which associates to an element $x$ in $L_0$ his class in $\overline{L_0}$ and $F_{1,-1}:L_{-1}\rightarrow \overline{L_{-1}'}$, which associates to an element $a=a'+c$ in $L_{-1}=\ker(l_1)\oplus C$, the element $a'$ in $\ker(l_1),$
as illustrated by the following diagram : 
$$\begin{matrix}       
&L_{-1} = \ker(l_1) \oplus C &\stackrel{l_1}{\longrightarrow} & L_0= \mathrm{im}(l_1)\oplus C'&\\
&&&&\\
&\Bigg\downarrow F_{1,-1} &  & \Bigg\downarrow F_{1,0}&\\
&&&&\\
&\overline{L_{-1}}= \ker(l_1) &\stackrel{\overline{l_1}=0}{\longrightarrow}&\overline{L_0}= L_0/\mathrm{im}(l_1) \, .&\\

\end{matrix}$$
By construction, $F_1$ is a morphism of chain complexes from 
$L_{-1}\stackrel{l_1}{\longrightarrow} L_{0}$ to a chain complex with trivial differential
$\overline{L_{-1}}\stackrel{\overline{l_1}=0}{\longrightarrow} \overline{L_{0}}$. This 
morphism clearly induces an isomorphism of the associated cohomology vector spaces.
Now, the Lie 2-algebra morphism $(F_1,F_2)$ has to satisfy the following relations.
\begin{itemize}
        \item [$(A_1)$] $F_{1,0}\circ l_1=0$
        \item [$(A_2)$] $\overline{l_2^p}(F_{1,0}(x),F_{1,0}(y))=F_{1,0}(l_2^p(x,y))$
        \item [$(A_3)$] $F_{1,-1}(l_2^m(a,x))=F_2(l_1(a),x)+
        \overline{l_2^m}(F_{1,-1}(a),F_{1,0}(x))$
       \item [$(A_4)$] $F_{1,-1}(l_3(x,y,z))+\big{(}F_2(l_2^p(x,y),z)\big{)}+ c.p.$\\
       $\quad \quad =\overline{l_3}(F_{1,0}(x),F_{1,0}(y),F_{1,0}(z))+\big{(}\overline{l_2^m}(F_{1,0}(x),F_2(y,z))\big{)}+c.p.$
       \par\noindent
       \end{itemize}
       where $c.p.$ denotes the cyclic permutation.\\
       The relation $(A_1)$ is satisfied by construction. The relation $(A_2)$ implies the above formula for $\overline{l_2^p}$, that is immediately checked to be well-defined. We define $\overline{l_2^m}$, for $a'$ in $\ker(l_1)$ and $x$ in $L_0$, by 
       $$\overline{l_2^m}(a',F_{1,0}(x))= l_2^m(a,x),$$ imposed by $(A_3)$
       and again easily checked to be well-defined.
       For $a'$ in $\ker(l_1)$, $c$ in $C$, $x$ in $L_0$
       and $b$ in $L_{-1}$, the following identity must be satisfied
       $$\overline{l_2^m}\big(F_{1,-1}(a'+c), F_{1,0}(x+l_1(b))\big)= \overline{l_2^m}\big(a',F_{1,0}(x)\big).$$
       This relation implies that for all $c$ in $C$ and $x$ in $L_0$, 
       $$F_2(x,l_1(c))=F_{1,-1}(l_2^m(x,c)).$$
       We thus define $F_2$ for all $x$ and $y$ in $L_0$, by
       $$F_2(x,y)=F_{1,-1}(l_2^m(x,c))$$
       where we have decomposed $y$ as $$y=l_1(a'+c)+c'=l_1(c)+c'$$
       with $a'$ in $\ker(l_1)$, $c$ in $C$ and $c'$ in $C'$, using the decomposition of $L_0$ and $L_{-1}$.
       Finally,
       the relation $(A_4)$ gives the definition of $\overline{l_3}$. For all $x$,$y$,$z$ in $L_0$,
       $$\overline{l_3}(F_{1,0}(x),F_{1,0}(y),F_{1,0}(z))= F_{1,-1}(l_3(x,y,z))+\big(F_2(l_2^p(x,y),z)\big)+c.p.-\big(\overline{l_2^m}(F_{1,0}(x), F_2(y,z))\big)+c.p.$$
       Since the mappings $\overline{l_2}$ and $\overline{l_3}$ are well-defined, we conclude that 
$(\overline{L_{-1}}\oplus\overline{L_0}, \overline{l_1}=0, \overline{l_2}, \overline{l_3})$ is a skeletal Lie 2-algebra. 
\end{proof}

\subsection{Examples}\hfill

\medskip

Let us check explicitly the above quasi-isomorphism for two fundamental  examples of Lie 2-algebras.

\bigskip

\noindent (1) For the algebra of observables of Rogers of $(M,\omega)$, a connected 2-plectic manifold, the quasi-isomophism is described by the following diagram
$$\begin{matrix}       
&L_{-1} =C^{\infty}(M) &\stackrel{l_1=d}{\longrightarrow} & L_0=\Omega^1_{Ham}(M)&\\
&&&&\\
&\Bigg\downarrow F_{1,-1} &  & \Bigg\downarrow F_{1,0}&\\
&&&&\\
&\overline{L_{-1}}= \ker(l_1)\cong \R &\stackrel{\overline{l_1}=0}{\longrightarrow}&\overline{L_0}= L_0/\mathrm{im}(l_1)=L_0/{dC^{\infty}(M)}\, .&\\

\end{matrix}$$
In this case $l_2^m=0$, $F_2=0$. Furthermore, $F_{1,0}$ associates to a 1-form $\a$ in $\Omega^1_{Ham}(M)$ its class in $L_0/dC^{\infty}(M)$, $F_{1,-1}$ associates to a function $f$ in $C^{\infty}(M)$ an element of $\R$ and for all $\alpha$, $\beta$, $\gamma$ in $L_0$,
$$\overline{l_2}^p(F_{1,0}(\a),F_{1,0}(\b))= F_{1,0}(\iota_{X_\a\wedge X_{\b}}(\omega)) \,\, \hbox{and}$$
$$\overline{l_3}(F_{1,0}(\a), F_{1,0}(\b), F_{1,0}(\gamma)= F_{1,-1}\big(\iota_{X_\a\wedge X_\b\wedge X_{\gamma}}(\omega)\big).$$
\vskip0,25cm
\noindent (2) Let the Lie 2-algebra $( \mathfrak{g_{-1}}\oplus\mathfrak{g}_{0}, l_1, l_2, l_3)$ be given by $\mathfrak{g_{-1}}=<a>_{\R}$ and $\mathfrak{g}_0=<x_1,x_2,x_3>_{\R}$ and
the following brackets
\begin{itemize}
    \item $l_1(a)=x_2$,
     \item $l_2(x_1,x_2)=-x_2$,
     \item $l_2(x_1,x_3)=x_1$,
     \item $l_2(a,x_1)=a$,
     \item $l_3(x_1,x_2, x_3)=-a$.
\end{itemize}

For this Lie 2-algebra, one obtains the below diagram 
$$\begin{matrix}       
&\mathfrak{g}_{-1} =<a> &\stackrel{l_1}{\longrightarrow} & \mathfrak{g}_0&\\
&&&&\\
&\Bigg\downarrow F_{1,-1} &  & \Bigg\downarrow F_{1,0}&\\
&&&&\\
&\overline{\mathfrak{g}_{-1}}=\{0\} &\stackrel{\overline{l_1}=0}{\longrightarrow}&\overline{\mathfrak{g}_0}= \mathfrak{g}_0/\mathrm{im}(l_1)=
\mathfrak{g}_0/<x_2>_{\R} &\\
\end{matrix}$$
In this case, $F_2=0$, $l_2^m=0$, $l_3=0$, $F_{1,-1}=0$ and $\overline{\mathfrak{g}_0}=<F_{1,0}(x_1),F_{1,0}(x_3)>$
is the Lie algebra with the bracket
$$\overline{l_2^p}(F_{1,0}(x_1), F_{1,0}(x_3))= F_{1,0}(x_1).$$


\section{2-actions and endomorphisms of the space of observables}\label{endo}

Since the Lie algebra structure on vector fields on a manifold $M$ is induced by the commutator of endomorphisms of the "observables" $C^{\infty}(M)$, in this appendix we extend this mechanism to the case of the Lie 2-algebra $\mathfrak{X}^2(M)\oplus\mathfrak{X}^1(M)$.


\subsection{The Lie 2-algebra structure on the endomorphisms of a 2-graded vector space}


\noindent Let $V=V_{-1}\oplus V_0$ be a graded vector space, $V_{-1}$ being of degree $-1$ and $V_0$ of degree $0$.
We consider the space $\End(V_{-1}\oplus V_0)=Hom(V_0, V_{-1})\oplus \left(\End (V_{-1})\times \End (V_0) \right)$ as graded vector space equipped with the Lie 2-algebra structure (\cite{LaLi2}):

\begin{enumerate}
    
 \item $\tau_1: Hom(V_0, V_{-1})\rightarrow \End (V_{-1})\times \End (V_0)$ and $\tau_1=0$, 
     \item $\tau_2$ splits into
     \begin{enumerate}
        
          \item $\tau_2^p:\big(\End (V_{-1})\times \End (V_0)\big)\times \big(\End (V_{-1})\times \End (V_0)\big)\rightarrow \End (V_{-1})\times \End (V_0)$, defined by $\tau_2^p((s_{-1},s_0),(t_{-1}, t_0))=([s_{-1},t_{-1}], [s_0,t_0])$ where $[.,.]$ is the commutator of endomorphisms when $s_i,t_i\in \End(V_i) $ for $i\in\{-1,0\}$.
     \item $\tau_2^m: Hom(V_0, V_{-1})\times \big(\End (V_{-1})\times \End (V_0)\big)\rightarrow Hom(V_0, V_{-1})$ defined by  $\tau_2^m(\phi,(s_{-1}, s_0))=\phi\circ s_0-s_{-1}\circ \phi$ for $\phi\in Hom(V_0, V_{-1})$, $s_0\in \End(V_0)$ and $s_{-1}\in\End(V_{-1})$ 
     \end{enumerate}
     \item $\tau_3: \bigwedge^3 \big(\End ( V_{-1})\times \End (V_0) \big)\rightarrow Hom(V_0, V_{-1})$ and $\tau_3=0$.
 \end{enumerate}
\begin{proof}
As $\tau_1=0$ and $\tau_3=0$, the relation $(R_1)$, $(R_2)$, $(R_3)$ and $(R_6)$ are satisfied. $(R_4)$ follows from the Jacobi identities on $\End (V_{-1})$ and $\End (V_0)$, and $(R_5)$ from direct verification.
\end{proof}
\begin{rema}
    We suppress here the term $Hom(V_{-1},V_0)$ since it would be of degree $1$.
\end{rema}
\bigskip


\subsection{2-actions as endomorphisms}\hfill\\
Consider the map
$F: \big(\mathfrak{X}^2(M)\oplus\mathfrak{X}^1(M), 0,\nu_2, 0\big)\rightarrow \left(\End\left(C^\infty(M) \oplus \Omega^1(M)\right), 0, \tau_2,0 \right)$ given by:
\begin{enumerate}
    \item $F_{1,0}:\mathfrak{X}^1(M)\rightarrow \End (C^\infty(M))\times \End (\Omega^1(M))$, $X\mapsto (\mathfrak{L}_X,\mathfrak{L}_X)$
    \item $F_{1,-1}:\mathfrak{X}^2(M)\rightarrow Hom(\Omega^1(M), C^\infty(M))$, $X\wedge Y\mapsto \mathfrak{L}_{X\wedge Y}$
     \item $F_2:\bigwedge^2(\mathfrak{X}^1(M))\rightarrow Hom(\Omega^1(M), C^\infty(M))$, $F_2=0$.
    \end{enumerate}

\begin{prop}
     $F$ is a Lie 2-morphism.
\end{prop}
\begin{proof}
The relations $(A_1)$ and $(A_4)$ are obvious and $(A_2)$ follows directly from $[\mathfrak{L}_X,\mathfrak{L}_Y]=\mathfrak{L}_{[X,Y]_S}$.
It is actually the same for relation $(A_3)$:  $\mathfrak{L}_{[X\wedge Y,Z]_S} = [\mathfrak{L}_{X\wedge Y} , \mathfrak{L}_Z]$ and the latest is $\tau_2^m\left( \mathfrak{L}_{X\wedge Y}, (\mathfrak{L}_Z,\mathfrak{L}_Z)\right)$ which gives $(A_3)$ ($\displaystyle F_{1,-1}([X\wedge Y,Z]_S) = \tau_2^m\left(F_{1,-1}(X\wedge Y), (\mathfrak{L}_Z,\mathfrak{L}_Z)\right)$.
\end{proof}

\begin{rema}
Obviously, the above map $F$ is a natural generalisation of the Lie derivative of vector fields.
The linear map $F_1$ is injective since $F_{1,0}$ and $F_{1,-1}$  are injective.
\end{rema}



\section{Examples of 2-actions and comomentum maps}\label{app-ex}

We consider Lie 2-algebras $(\mathfrak{g}= \mathfrak{g}_{-1}\oplus\mathfrak{g}_{0}, l_1, l_2, l_3)$ and give several explicit examples of Lie 2-actions and comomentum maps, covering different cases, classified via the below shorthands :\\
\noindent S0  signals that $\mathfrak{g}_{0}$ is not a Lie algebra, equivalently $l_1\circ l_3\neq 0$ (and thus $l_1\neq 0$ and $l_3\neq 0$), and \\
\noindent $(\neg$ S0) signals that $\mathfrak{g}_{0}$ is  a Lie algebra, i.e., $l_1\circ l_3 = 0$.\\
    
\noindent We refine the cases (also called situations in this appendix) using the following conditions on the Lie 2-algebra structures (Sx) resp. on the 2-action (Tx) :    
\begin{multicols}{2}
    \begin{enumerate} 
    \item[S1:] $l_1\neq 0$ 
    \item[S2:] $l_1=0$
    \item[S3:] $l_2\neq 0$
    \item[S4:] $l_2= 0$
    \item[S5:] $l_3\neq 0$
    \item[S6:] $l_3= 0$
\end{enumerate}

\begin{enumerate}
    \item[T1:] $\rho_{1,0}\neq 0$ 
    \item[T2:] $\rho_{1,0}= 0$
    \item[T3:] $\rho_{1,-1}\neq 0$ 
    \item[T4:] $\rho_{1,-1}= 0$
    \item[T5:] $\rho_{2}\neq 0$ 
    \item[T6:] $\rho_{2}= 0$
\end{enumerate}
\end{multicols}

\newpage

\begin{tabular}{|p{1,7cm}|p{1,5cm}|p{1,2cm}|p{1,2cm}|p{1,2cm}|p{1,2cm}|p{1,2cm}|p{1,2cm}|p{1,2cm}|p{1,2cm}|} 
\hline  
  &  S0 & $(\neg$ S0) &$(\neg$ S0) & $(\neg$ S0)&$(\neg$ S0)& $(\neg$ S0)&$(\neg$ S0)&$(\neg$ S0) &$(\neg$ S0)\\  
\hline
 & \quad\par $l_1\circ l_3\neq 0$ &S135\par $l_1\neq 0$,\par $l_2\neq 0$, \par $l_3\neq 0$ & S136\par $l_1\neq 0$,\par $l_2\neq 0$, \par $l_3= 0$&S145\par $l_1\neq 0$,\par $l_2=0 $, \par $l_3\neq 0$& S146\par $l_1\neq 0$,\par $l_2= 0$, \par $l_3= 0$&S235\par $l_1= 0$,\par $l_2\neq 0$, \par $l_3\neq 0$&S236\par $l_1= 0$,\par $l_2\neq 0$, \par $l_3= 0$&S245\par $l_1= 0$,\par $l_2= 0$, \par $l_3\neq 0$&S246\par $l_1= 0$,\par $l_2= 0$, \par $l_3= 0$ \\
\hline 
T135\par $\rho_{1,0}\neq 0$\par $\rho_{1,-1}\neq 0$\par $\rho_{2}\neq 0$  &  \quad \par \centering Ex.1& \quad \par \centering Ex.2&\quad \par \centering Ex.2a \par Ex.3b \par & \quad \par \centering Ex.2a&\quad \par \centering Ex.2a& \quad \par \centering Ex.3a&\quad \par \centering Ex.2a& &\\  
\hline
T136\par $\rho_{1,0}\neq 0$\par $\rho_{1,-1}\neq 0$\par $\rho_{2}= 0$  & \quad \par \centering Ex.1a&  \quad \par \centering Ex.2a&\quad \par \centering Ex.2a \quad \par \centering Ex.3b & \quad \par \centering Ex.2a&\quad \par \centering Ex.2a&\quad \par \centering Ex.2b&\quad \par \centering Ex.2a \par Ex.4b& &\\  
\hline
T145\par $\rho_{1,0}\neq 0$\par $\rho_{1,-1}= 0$\par $\rho_{2}\neq 0$  & \quad \par \centering Ex.1&\quad \par \centering Ex.2a\  &\quad \par \centering  Ex.2a \par Ex.3b &\quad \par \centering Ex.2a &\quad \par \centering Ex.2a&&\quad \par \centering Ex.2a \par Ex.4a & \quad \par \centering Ex.2b &\\  
\hline
T146\par $\rho_{1,0}\neq 0$\par $\rho_{1,-1}= 0$\par $\rho_{2}= 0$  &  \quad \par \centering Ex.1a &\quad \par \centering Ex.2a &\quad \par \centering Ex.2a \par   Ex.3b&\quad \par \centering Ex.2a &\quad \par \centering Ex.2a& \quad \par \centering Ex.3a&\quad \par \centering Ex.2a \par Ex.4b &\quad \par \centering Ex.2b &\\  
\hline
T235\par $\rho_{1,0}= 0$\par $\rho_{1,-1}\neq 0$\par $\rho_{2}\neq 0$  &  \quad \par \centering Ex.1&\quad \par \centering Ex.2a  & \quad \par \centering Ex.2a \par Ex.3b &\quad \par \centering Ex.2a &\quad \par \centering Ex.2a  & \quad \par \centering Ex.3a&\quad \par \centering Ex.2a& &\\  
\hline
T236\par $\rho_{1,0}= 0$\par $\rho_{1,-1}\neq 0$\par $\rho_{2}= 0$  & &\quad \par \centering Ex.2a &\quad \par \centering Ex.2a \par   Ex.3b& \quad \par \centering Ex.2a&\quad \par \centering Ex.2a&&\quad \par \centering Ex.2a\par Ex.4b &&  \quad \par \centerline{Ex.5}  \\   
\hline
T245\par $\rho_{1,0}= 0$\par $\rho_{1,-1}= 0$\par $\rho_{2}\neq 0$  & \quad \par \centering Ex.1&\quad \par \centering Ex.2a  & \quad \par \centering Ex.2a \par Ex.3b &\quad \par \centering Ex.2a &\quad \par \centering Ex.2a&&\quad \par \centering Ex.2a\par Ex.4a& \quad \par \centering Ex.2b&\\   
\hline
T246\par $\rho_{1,0}= 0$\par $\rho_{1,-1}= 0$\par $\rho_{2}= 0$  & \quad \par \centering Ex.1a &\quad \par \centering Ex.2a  &\quad \par \centering Ex.2a \par Ex.3b  &\quad \par \centering Ex.2a &\quad \par \centering Ex.2a&&\quad \par \centering Ex.2a\par Ex.4&\quad \par \centering Ex.2b &\\   
\hline
\end{tabular}
\vskip1cm

\noindent Of course, finding the entry "Ex.2a" in the column S145 and the line T235 corresponds to the fact that the Example 2a below satisfies the following conditions : S1, S4 and S5
for the Lie 2-algebra and T2, T3 and T5 for the 2-action (as should be clear from the inscriptions in the table). \\
\noindent Note that a given example can appear in several entries of the table due to the presence of variants of it.
\newpage

\textbf{List of examples}\\

We give the examples in decreasing order of algebraic complexity of the brackets $l_1, l_2, l_3$, starting with  interesting 2-actions ``far'' from Lie algebra actions. 
\medskip

In the below list the manifold $M$ will mostly be an open subset of $\mathbb{\R}^m$ with coordinates $(q^1, \ldots, q^m)$.
Furthermore we always write $\partial_k$ for the vector field $\frac{\partial}{\partial q^k}$ with $k=1,\ldots, m$. If a construction generalizes to larger classes of manifolds without further effort, we will formulate it in the more general setting.

\medskip

For the description of a 2-action there is, of course, no need for a 2-plectic form on $M$. When we consider actions with comomentum map ${\lambda}$, we only consider fundamental 2-comomentum maps (except in Example 5) and  we will take the 2-plectic manifold $(M,\omega)=(\mathbb{R}^3,\o)$, with $\omega=dq^1\wedge dq^2 \wedge dq^3$,
the standard volume form on $\mathbb{R}^3$, except for the Examples 2b, 2c, 4b and 5.
\vskip1cm

\subsection*{Example 1} Lie 2-algebras  $\mathfrak{g}$ satisfying $l_1\circ l_3 \neq 0$ ($\mathfrak{g}_{0}$ is not a Lie algebra)\\

In this example, which covers the situation $S0$, we construct several Lie 2-actions of a given Lie 2-algebra on $M$. We obtain conditions for the existence of a 2-comomentum map when $M=\mathbb{R}^3$ and, when it exists, we specify the comomentum maps. 
\medskip

\subsubsection*{\underline{Example 1a}:}

Let us consider the Lie 2-algebra $(\mathfrak{g}= \mathfrak{g}_{-1}\oplus\mathfrak{g}_{0}, l_1, l_2, l_3)$ given by $\mathfrak{g}_{-1}=\langle a\rangle_{\mathbb{R}}$ and $\mathfrak{g}_0=\langle x_1,x_2,x_3\rangle_{\mathbb{R}}$ with the following non-vanishing brackets:
\begin{enumerate}
    \item $l_1(a)=x_2$
     \item $l_2(x_1,x_2)=-x_2$
     \item $l_2(x_1,x_3)=x_1$
     \item $l_2(a,x_1)=a$
     \item $l_3(x_1,x_2, x_3)=-a$.
\end{enumerate}

\noindent A Lie 2-action of $\mathfrak{g}$ on a manifold $M$, (i.e., a Lie 2-morphism $\rho$ from $\mathfrak{g}_{-1}\oplus\mathfrak{g}_{0}$ to $\mathfrak{X}^2(M)\oplus\mathfrak{X}^1(M)$) must satisfy: 
 \begin{enumerate}
     \item[$(A_1)$] $\rho_{1,0}(x_2)=0$ 
     \item[$(A_2)$] $\rho_{1,0}(x_1)=[\rho_{1,0}(x_1),\rho_{1,0}(x_3)]_S$
     \item[$(A_{31})$] $\rho_{1,-1}(a)=[\rho_{1,-1}(a),\rho_{1,0}(x_1)]_S + \rho_2(x_2,x_1)$
      \item[$(A_{32})$] $[\rho_{1,-1}(a),\rho_{1,0}(x_3)]_S + \rho_2(x_2,x_3)=0$
      \item[$(A_4)$] $-\rho_2(x_2,x_3)-\rho_2(x_1,x_2)-\rho_{1,-1}(a)=[\rho_{1,0}(x_1),\rho_2(x_2,x_3)]_S+[\rho_{1,0}(x_3),\rho_2(x_1,x_2)]_S$.
\end{enumerate}

\begin{rema} \hfill
\
\begin{itemize}
    \item The  relation $(A_4)$ is automatically satisfied because of the relations $(A_2), (A_{31})$ and $(A_{32})$.
    \item There is no condition on $\rho_2(x_1,x_3)$. So we don't have to specify $\rho_2(x_1,x_3)$ in the below sub examples.
    \item $\rho_2=0 $ implies $ \rho_{1,-1}=0$.
\end{itemize}
\end{rema}

\bigskip

\textbf{2-actions with $\rho_{1,0}(x_1)=0$}
\begin{multicols}{2}
\begin{enumerate}
    \item[] Sub example 1:
\begin{itemize}
    \item $\rho_{1,0}(x_1)=\rho_{1,0}(x_2)=0$
     \item $\rho_{1,0}(x_3)=\partial_3$
     \item $\rho_{1,-1}(a)=\rho_2(x_2,x_1)=\partial_2\wedge\partial_3$
      \item $\rho_2(x_2,x_3)=0$.
      \end{itemize}

    \item [] Sub example 2:
\begin{itemize}
    \item $\rho_{1,0}(x_1)=\rho_{1,0}(x_2)=0$
     \item $\rho_{1,0}(x_3)=\partial_3$
     \item $\rho_{1,-1}(a)=\rho_2(x_2,x_1)=\partial_1\wedge\partial_2$
      \item $\rho_2(x_2,x_3)=0$. 
\end{itemize}
\end{enumerate}
\end{multicols}

\textbf{Action on a 2-plectic manifold}

Let now $(M,\omega)$ be the 2-plectic manifold
$(\mathbb{R}^3, dq^1\wedge dq^2 \wedge dq^3$) and consider the existence question for a comomentum map.\\

First of all, we must fulfill the following conditions:
\begin{itemize}
    \item $\iota_{\rho_{1,0}(x_1)}\omega$ and $\iota_{\rho_{1,0}(x_3)}\omega$ are closed
    \item $\iota_{\rho_{1,-1}(a)}\omega$ is closed 
    \item $\iota_{\rho_2(x_i,x_j)}\omega$ is  closed $\forall i,j$.
\end{itemize}
\medskip

\textbf{Explicit 2-comomentum maps}

The above sub examples 1 and 2 with $\rho_{1,0}(x_1)=0$ satisfy 
the three closedness conditions if we add $\iota_{\rho_2(x_1,x_3)}\omega$ being closed.
The ensuing formulas for ${\lambda}$ are then:

\begin{enumerate}
   \item ${\lambda}_{1,0}(x_1)={\lambda}_{1,0}(x_2)=0$ 
     \item ${\lambda}_{1,0}(x_3)=q_1dq_2$
      \item ${\lambda}_2(x_1,x_3)=(c(x_1,x_3), (f,{\rho_2(x_1,x_3)}))$, where $c\in\bigwedge^2(\mathfrak{g}_0^*)$ and $\iota_{\rho_2(x_1,x_3)}\omega= -d f_{\rho_2(x_1,x_3)}$
      \item ${\lambda}_{1,-1}(a)= {\lambda}_2(x_2,x_1)$
      \item \hbox{and for sub example 1 }
      
     $ \begin{cases}
{\lambda}_2(x_2,x_1)=(0,(q_1,\partial_2\wedge\partial_3)) \\
{\lambda}_2(x_2,x_3)=0 \\
\end{cases}$
\item \hbox{or for sub example 2 }

$\begin{cases}
{\lambda}_2(x_2,x_1)=(0,(q_3,\partial_1\wedge\partial_2)) \\
{\lambda}_2(x_2,x_3)=(0,(1,0)) \\
        \end{cases}$
\end{enumerate}

Observe that in this example the existence of a  fundamental comomentum map depends on $\iota_{\rho_2(x_1,x_3)}\omega$.
\begin{rema}
    The  above formulas for the sub example 2 is imposed by the conditions $(A_3)$ and $(A_4)$ for ${\lambda}$.
\end{rema}

\textbf{2-action with $\rho_{1,0}(x_1)\neq 0$}

In this case, the following data satisfy $(A_1)-(A_4)$:

\begin{enumerate}
   \item [(a)] $\rho_{1,0}(x_1)=e^{-q_i}\partial_j$ with $1\leq i\neq j \leq m$
     \item [(b)]$\rho_{1,0}(x_3)=\partial_i$
     \item [(c)]$\rho_{1,-1}(a)=\rho_2(x_2,x_1)=\partial_k\wedge\partial_l$ with $k,l\neq i$
      \item [(d)]$\rho_2(x_2,x_3)=0$. 
\end{enumerate}

\begin{rema}
    The choices $(c)$ and $(d)$ above give 2-actions also in the case $\rho_{1,0}= 0$, and the choices $(a)$ and $(b)$ give 2-actions when $\rho_{1,-1}=0$ and $\rho_2 = 0$.
    Thus the 2-actions of Example 1 cover the situations $T235$, $T135$ and $T4$.
\end{rema}
\begin{rema}
If we add an element to $\mathfrak{g}_{-1}$, say $\mathfrak{g}_{-1}=\langle a, b\rangle_{\mathbb R}$ with the same brackets as above then we can construct a Lie 2-action with $\rho_2=0$ and $\rho_{1,-1}\neq 0$ since $\rho_{1,-1}(b)$ is not related to $\rho_2$ and must only satisfy the condition $[\rho_{1,-1}(b),\mathrm{im}(\rho_{1,0})]= 0$. Hence we are in situation $T136$.
\end{rema}


\bigskip

\subsubsection*{\underline{Example 1b}:} 

Example of a Lie 2-algebra with $dim(\mathfrak{g}_{-1})=2$, $l_2\neq 0$ and $l_1\circ l_3\neq 0$.

Let us consider the Lie 2-algebra $(\mathfrak{g}= \mathfrak{g}_{-1}\oplus\mathfrak{g}_{0}, l_1, l_2, l_3)$ given by $\mathfrak{g}_{-1}=\langle a_z,a_t\rangle_{\mathbb{R}}$ and \\ $\mathfrak{g}_0=\langle x,y,z,t\rangle_{\mathbb{R}}$ with the relations:
\begin{itemize}
    \item $l_1(a_z)=z,l_1(a_t)=t$
     \item $l_2(x,y)=-y$, $l_2(y,z)=z$, $l_2(y,t)=t$
     \item $l^m_2(y,a_z)=a_z$, $l^m_2(y,a_t)=a_t$
     \item $l_3(x,y,z)=a_z$, $l_3(x,y,t)=a_t$.
\end{itemize}

Then a Lie 2-action $\rho$ of this Lie 2-algebra on $M=\mathbb{R}^3$ is given by:
\begin{itemize}
    \item $\rho_{1,0}(y)=0$, $\rho_{1,0}(z)=0$ $\rho_{1,0}(t)=0$
    \item $\rho_{2}(z,t)=0$
    \item  any choice of $\rho_{1,0}(x)$, $\rho_{1,-1}(a_z)$ and $\rho_{1,-1}(a_t)$
    \item any choice of $\rho_{2}(x,y)$
    \item other values of $\rho_{2}$ fixed by
    \begin{itemize}
    \item[(i)] $\rho_{2}(x,z)=[\rho_{1,-1}(a_z),\rho_{1,0}(x)]$, $\rho_{2}(x,t)=[\rho_{1,-1}(a_t),\rho_{1,0}(x)]$
    \item[(ii)] $\rho_{2}(y,z)=\rho_{1,-1}(a_z)$
    $\rho_{2}(y,t)=\rho_{1,-1}(a_t)$,
    \end{itemize}
\end{itemize}
since it is enough to fix $\rho_{1,0}(x)$, $\rho_{1,-1}(a_z)$ and $\rho_{1,-1}(a_t)$ to determine $\rho_{2}(x,z)$, $\rho_{2}(x,t)$, $\rho_{2}(y,z)$ and $\rho_{2}(y,t)$.

\begin{rema}
As there is no condition on $\rho_{2}(x,y)$, the closedness conditions are not always satisfied. Thus certain actions cannot be lifted to a fundamental comomentum map.
\end{rema}

\noindent This covers situations $T135$, $T145$, $T235$, $T245$.

\bigskip

\bigskip

\subsection*{Example 2} Lie 2-algebras  $\mathfrak{g}$ satisfying $l_1\circ l_3 = 0$ ($\mathfrak{g}_{0}$ is a Lie algebra)\\

Here we cover all the situations $S$ imposed on the structures, (except $S0$, see Remark \ref{touslesS} ), and all the situations $T$ for the 2-actions (see Remark  \ref{touslesT}).
We provide an elementary construction of Lie 2-algebras (called basic Lie 2-algebras, Definition \ref{basic}) with $l_1\neq 0,$ $l_2\neq 0$ and $l_3\neq 0$.
The first example (Example 2a) uses an invariant bi-vector. The second and third examples (Examples 2b and 2c) give two 2-actions of basic Lie 2-algebras, on $(\mathbb{R}^6, \omega=dq_1\wedge dq_5\wedge dq_6-dq_2\wedge dq_4\wedge dq_6+ dq_3\wedge dq_4\wedge dq_5)$, the first one is not Hamiltonian, see Remark \ref{notHam} whereas the second is.\\

\textbf{Elementary construction }

\noindent Let us consider two vector spaces $\mathfrak{g}_{-1}$ and $\mathfrak{g}_0$, with a collection of three linear maps $(l_1,l_2,l_3)$, where
   $l_1:\mathfrak{g}_{-1}\mapsto \mathfrak{g}_0$, $l_2$ can be decomposed in 
   $l_2^p:\bigwedge^2 \mathfrak{g}_0\mapsto \mathfrak{g}_0$ and
   $l_2^m :\mathfrak{g}_{-1}\times \mathfrak{g}_0\mapsto \mathfrak{g}_0$, and finally
   $l_3 : \mathfrak{g}_0\times \mathfrak{g}_0\times \mathfrak{g}_0\mapsto \mathfrak{g}_{-1}$.
   Let us also assume that we have the following decompositions:
   $\mathfrak{g}_0=\mathfrak{g}_{0,0}\oplus \mathrm{im}(l_1)\oplus \mathrm{im}(l_2^p)  $ and $\mathfrak{g}_{-1}=\mathfrak{g}_{-1,0} \oplus \mathrm{im}(l_3)\oplus \mathrm{im}(l_2^m)  $.
   For convenience, we denote by $\mathfrak{g}_{0,2}=\mathrm{im}(l_2^p)$, $\mathfrak{g}_{0,1}=\mathrm{im}(l_1)$ and by $\mathfrak{g}_{-1,2}=\mathrm{im}(l_2^m)$, $\mathfrak{g}_{-1,3}=\mathrm{im}(l_3)$ thus $$\mathfrak{g}_0=\mathfrak{g}_{0,0}  \oplus \mathfrak{g}_{0,1}
   \oplus \mathfrak{g}_{0,2}  \text{ and } \mathfrak{g}_{-1}=\mathfrak{g}_{-1,0}\oplus \mathfrak{g}_{-1,3}\oplus \mathfrak{g}_{-1,2}.$$ 
   Then the condition that $(\mathfrak{g}=\mathfrak{g}_{-1}\oplus \mathfrak{g}_0, l_1,l_2,l_3)$ is a Lie 2-algebra is assured by  
   the following relations, for all $a, b$ in $\mathfrak{g}_{-1}$ and for all $x,y,z,t$ in $\mathfrak{g}_0$
   \begin{itemize}
 \item [$(R_2)$]  $l_1\vert_{\mathfrak{g}_{-1,2}}=0$ and 
 $l_2^p\vert_{\mathfrak{g}_{0,1}}=0$
   \item [$(R_3)$]  $l_2^m(l_1(a),b)=l_2^m(a,l_1(b))$
   \item [$(R_4)$] $l_1\vert_{\mathfrak{g}_{-1,3}}=0$ and $\mathfrak{g}_0$ is a Lie algebra
   \item [$(R_5)$] $l_3\vert_{\mathfrak{g}_{0,1}}=0$ and $l_2(l_2(x,y),a)+l_2(l_2(y,a),x)+l_2(l_2(a,x),y)=0$
   \item [$(R_6)$] $l_3(l_2(x,y),z,t)-l_3(l_2(x,z),y,t)+ l_3(l_2(x,t),y,z)+l_3(l_2(y,z),x,t)-l_3(l_2(y,t),x,z)+ l_3(l_2(z,t),x,y)=0$ and 
  $l_2(l_3(x,y,z),t)-l_2(l_3(x,y,t),z)+l_2(l_3(x,z,t),y)-l_2(l_3(y,z,t),x)=0$,
   \end{itemize}
where  the notation $l_1\vert_{\mathfrak{g}_{-1,2}}=0$ means that  $l_1(a)=0$, $\forall a\in \mathfrak{g}_{-1,2}$ and the notation  $l_2^p\vert_{\mathfrak{g}_{0,1}}=0$ means that  $l_2^p(x,l_1(a))=0$, $\forall a\in \mathfrak{g}_{-1}$ since $\mathfrak{g}_{0,1}=\mathrm{im}(l_1)$.

   \begin{prop}
  A sum of two vector spaces with linear maps $(\mathfrak{g}= \mathfrak{g}_{-1}\oplus\mathfrak{g}_{0}, l_1,l_2=(l_2^p,l_2^m),l_3)$ is a
  Lie 2-algebra if the following conditions are fulfilled :
  \begin{itemize}
     \item $\mathfrak{g}_{-1}=\mathfrak{g}_{-1,0}\oplus \mathfrak{g}_{-1,3}\oplus\mathfrak{g}_{-1,2}$ is a vector space
      \item $\mathfrak{g}_0=\mathfrak{g}_{0,0}\oplus\mathfrak{g}_{0,1}\oplus\mathfrak{g}_{0,2}$ is a Lie algebra
  \item the four following maps are linear:
  \begin{itemize}
      \item  $l_1:\mathfrak{g}_{-1,0}\mapsto \mathfrak{g}_{0,1}$
      \item $l_2^p:\bigwedge^2 (\mathfrak{g}_{0,0} \oplus \mathfrak{g}_{0,2} )\mapsto \mathfrak{g}_{0,2}$
      \item  $l_2^m :\mathfrak{g}_{-1,0}\times \mathfrak{g}_{0,0}\mapsto \mathfrak{g}_{-1,2}$ 
      \item  $l_3 : \bigwedge^3 \mathfrak{g}_{0,0}\mapsto \mathfrak{g}_{-1,3}$
      \item the remaining partial maps are zero.
  \end{itemize}
   \end{itemize}
\end{prop}
\begin{proof}
    
We use the above construction and apply the relations $(R_i)$ for $i=2,\ldots 6$.
\end{proof}
\begin{defn}\label{basic}
    We will call such a Lie 2-algebra $\mathfrak{g}$ a $\emph{basic}$ Lie 2-algebra.
\end{defn}
   \begin{rema}
       If $\mathfrak{g}$ is basic and either $\mathfrak{g}_{0,0}= 0$ or $\mathfrak{g}_{-1,2}= 0$, then $l_2^m = 0$.
   \end{rema}

   \bigskip

   \textbf{2-action of a basic Lie 2-algebra}
    
\noindent Let us consider a basic Lie 2-algebra. Then the following conditions on $\rho:\mathfrak{g}_0\oplus \mathfrak{g}_{-1}\longrightarrow \mathfrak{X}^2(M)\oplus\mathfrak{X}^1(M)$ assure that it is a Lie 2-action of   $\mathfrak{g}_0\oplus \mathfrak{g}_{-1}$ on the manifold $M$: 
   \begin{enumerate}
        \item $\rho_{1,0}$ is an action of $\mathfrak{g}_0$ and  $\rho_{1,0}\vert_{\mathfrak{g}_{0,1}}= 0$
       \item $\rho_{1,-1}:\mathfrak{g}_{-1,0}\longrightarrow \mathfrak{X}^2(M) $ and $\rho_{1,-1}\vert_{{\mathfrak{g}_{-1,2}}}=0$,
       $\rho_{1,-1}\vert_{{\mathfrak{g}_{-1,3}}}=0$
       \item $\rho_2: \bigwedge^2{\mathfrak{g}_{0,0}} \longrightarrow \mathfrak{X}^2(M) $ and $\rho_{2}\vert_{{\mathfrak{g}_{0,1}}}=0$, $\rho_{2}\vert_{{\mathfrak{g}_{0,2}}}=0$
       \item with 
        \begin{itemize}
           \item $[\rho_{1,0}(x),\rho_2(y,z)]_S +c.p.=0$ for $x,y,z\in \mathfrak{g}_{0,0}$
           \item $[\rho_{1,0}(x),\rho_2(y,z)]_S  =0$ for $x\in \mathfrak{g}_{0,2}$ and $y,z\in \mathfrak{g}_{0,0}$ 
           \item $[\rho_{1,0}(x),\rho_{1,-1}(a))]_S=0$.
       \end{itemize}
  \end{enumerate}

\medskip
  
\subsubsection*{\underline{Example 2a}:} 

We construct a Lie 2-action of a basic Lie 2-algebra on a manifold $M$ using an invariant bivector field.
    \begin{prop} \label{basicaction}
    Let $\mathfrak{g}_0$ be a basic Lie 2-algebra and assume that there exists
 \begin{itemize}
     \item an action $\rho_{1,0}$ of $\mathfrak{g}_0$ on $M$ such that  $\rho_{1,0}\vert_{\mathfrak{g}_{0,1}}= 0$
     \item a 2-tensor $\pi\in\mathfrak{X}^2(M)$ which is invariant by the action $\rho_{1,0}$.
 \end{itemize}
If $im(\rho_{1,-1}\vert_{{\mathfrak{g}_{-1,0}}})\subset \langle\pi\rangle_\mathbb{R}$, $\rho_{1,-1}\vert_{{\mathfrak{g}_{-1,2}}}=0$,
       $\rho_{1,-1}\vert_{{\mathfrak{g}_{-1,3}}}=0$ and 
  $im(\rho_{2}\vert_{{\mathfrak{g}_{0,0}}})\subset \langle\pi\rangle_\mathbb{R}$, $\rho_{2}\vert_{{\mathfrak{g}_{0,1}}}=0$, $\rho_{2}\vert_{{\mathfrak{g}_{0,2}}}=0$, then $\mathfrak{g}_0\oplus \mathfrak{g}_{-1}$ acts as a Lie 2-algebra on $M$. 
\end{prop}
\begin{proof}

Such a map $\rho$ satisfies the above relations $(1)-(4)$.
\end{proof}
\begin{rema}\label{touslesT}
    The preceding proposition covers all the situations T. In particular, if we take $\rho_{1,0}= 0$, there is no condition on $\rho_{1,-1}$ and $\rho_{2}$ (Situation $T2$).
\end{rema}

 \begin{rema}\label{touslesS}
  
 Via the preceding construction one can also obtain all the situations $S$ (except $S0$). They all are realized as special cases of the Proposition \ref{basicaction} and cover all the situations $T$. For example: 
 
\begin{enumerate}
    \item[S236] means that $l_1=l_3=0$. In this case, $\mathfrak{g}_0=\mathfrak{g}_{0,0}\oplus\mathfrak{g}_{0,2}$, $\mathfrak{g}_{-1}=\mathfrak{g}_{-1,0}\oplus\mathfrak{g}_{-1,2}$ and     we obtain an action of the Lie 2-algebra $\mathfrak{g}$ on a manifold $M$ when $\rho_{1,0}$ is an action of $\mathfrak{g}_0$ on $M$,
$\rho_{1,-1}:\mathfrak{g}_{-1,0}\longrightarrow \mathfrak{X}^2(M) $ and $\rho_2: \bigwedge^2{\mathfrak{g}_{0,0}} \longrightarrow \mathfrak{X}^2(M) $, as above. 
\item[S146] means that $l_1\neq 0, l_2=0, l_3=0$. Thus we obtain that $\mathfrak{g}_0=\mathfrak{g}_{0,0}\oplus\mathfrak{g}_{0,1}$, $\mathfrak{g}_{-1}=\mathfrak{g}_{-1,0}$. 
\item[S145] means that $l_1\neq 0, l_2=0, l_3\neq0$. Thus we obtain that $\mathfrak{g}_0=\mathfrak{g}_{0,0}\oplus\mathfrak{g}_{0,1}$, $\mathfrak{g}_{-1}=\mathfrak{g}_{-1,0}\oplus\mathfrak{g}_{-1,3}$. 
\item[S136] means that $l_1\neq 0, l_2\neq0, l_3=0$. Thus we obtain that $\mathfrak{g}_{-1,3}= 0$. 

\end{enumerate}

\end{rema}

 \medskip

\subsubsection*{\underline{Example 2b}:} 
We give an explicit construction of a basic 2-action which is not Hamiltonian on $(\mathbb{R}^6, \omega=dq_1\wedge dq_5\wedge dq_6-dq_2\wedge dq_4\wedge dq_6+ dq_3\wedge dq_4\wedge dq_5)$ for the basic Lie 2-algebra (of the style "ax+b") given by:

\begin{enumerate}
    \item the Lie algebra $\mathfrak{g}_0$ is a direct sum of Lie algebras "$ax+b$", more precisely $\mathfrak{g}_0=\mathfrak{g}_{0,0}\oplus\mathfrak{g}_{0,2}\oplus\mathfrak{g}_{0,1}=\langle y_1,\ldots, y_{3l}\rangle_{\mathbb{R}}\oplus \langle x_1,\ldots, x_{3l}\rangle_{\mathbb{R}}\oplus \langle y_{3l+1},\ldots, y_{3l+k}\rangle_{\mathbb{R}}$ with $[x_i,y_i]_{\mathfrak{g}_0}=x_i$, $\forall i=1\ldots3l$ and all other brackets equal to zero
    \item the vector space $\mathfrak{g}_{-1}$ is given by $\mathfrak{g}_{-1}=\mathfrak{g}_{-1,3}\oplus\mathfrak{g}_{-1,0}\oplus \mathfrak{g}_{-1,2}=\langle a_1,\ldots, a_l\rangle_{\mathbb{R}}\oplus \langle b_1,\ldots, b_k\rangle_{\mathbb{R}}\oplus \langle c_{i,j} \, \vert\,  i=1,\ldots,k; j=1,\ldots,3l\rangle_{\mathbb{R}}$ 
    \item the linear maps are defined by
    \begin{itemize}
       \item $l_2^p=[.,.]_{\mathfrak{g}_0}$
        \item $l_3(y_{3i+1},y_{3i+2}, y_{3i+3})=a_{i+1} $ for $i=0,\ldots, l-1$
        \item $l_1(b_j)=y_{3l+j}$ for $j=1,\ldots, k$ 
        \item $l^2_m(b_j, y_i)=c_{i,j}$ for $i=1,\ldots,k; j=1,\ldots,3l$.
    \end{itemize}

\end{enumerate}
\noindent Then $\mathfrak{g}_0\oplus \mathfrak{g}_{-1}$ is a Lie 2-algebra which acts on an open set in $\mathbb{R}^{3l}$ by: 
\begin{itemize}
    \item $\rho_{1,0}(y_j)=\partial_j$, $\rho_{1,0}(x_j)=e^{-q_j}\partial_j$, for $j=1,\ldots,3l-2$
    \item $\rho_{1,0}(x_{3l-1})=\rho_{1,0}(x_{3l})=0=\rho_{1,0}(y_{3l+i})$ for $i=1\ldots, k$
    \item $\rho_2(y_i,y_j)=0$ for $i,j\neq 3l, 3l-1$
    \item $\rho_2(y_{3l-1},y_{3l})=\partial_{3l-1}\wedge\partial_{3l}$
    \item $\rho_{1,-1}(b_j)=\partial_{3l-1}\wedge\partial_{3l}$ for $j=1,\ldots, k$
    \item There is no condition on $\rho_{1,0}(y_{3l-1})$ and $\rho_{1,0}(y_{3l})$.
\end{itemize}

\begin{rema} \label{quasiisomorphism}
As $\rho_{1,-1}$ is only defined on $\mathfrak{g}_{-1,0}$ the quasi-isomorphism of Appendix A annihilates the action since  $\overline{{\rho}_{1,-1}}= 0$.
\end{rema}
\begin{rema}
      With this example we also obtain the situation $S245$ with $l_1=0, l_2=0,  l_3\neq 0$. In this case, $\mathfrak{g}_0=\mathfrak{g}_{0,0}$ and $\mathfrak{g}_{-1}=\mathfrak{g}_{-1,3}$ and $\rho_{1,-1}= 0$. We can take $\rho_{1,0}(y_j)=\partial_j$, for $j=1,\ldots,3l-2$ and $\rho_2(y_{3l-1},y_{3l})=\partial_{3l-1}\wedge\partial_{3l}$. This covers situation  $T4$.
  \end{rema}
\begin{rema}\label{notHam}
     If we consider the Lie 2-action on the 2-plectic manifold $(\mathbb{R}^6, \omega=dq_1\wedge dq_5\wedge dq_6-dq_2\wedge dq_4\wedge dq_6+ dq_3\wedge dq_4\wedge dq_5)$, then this Lie 2-action cannot be lifted to a fundamental comomentum map since for $j=1, ..., 4$ one has $d(\iota_{\rho_{1,0}(x_j)}\omega)\neq 0$, i.e.,  the closedness condition is not satisfied and thus there exists no ${\lambda}_{1,0}(x_j)$.
\end{rema}

\medskip

\subsubsection*{\underline{Example 2c}:} We give an explicit construction of a Hamiltonian 2-action  on $(\mathbb{R}^6, \omega=dq_1\wedge dq_5\wedge dq_6-dq_2\wedge dq_4\wedge dq_6+ dq_3\wedge dq_4\wedge dq_5)$ for the basic Lie 2-algebra (of style "ax+b") given by:
\begin{itemize}
     \item $k\geq 1$
    \item $\mathfrak{g}_0=\langle y_1,\ldots ,y_{6}\rangle_{\mathbb{R}}\oplus\langle x_1,\ldots, x_{6}\rangle_{\mathbb{R}}\oplus \langle y_{7},\ldots, y_{6+k}\rangle_{\mathbb{R}}$ with $[x_i,y_i]=x_i$, $\forall i=1,\ldots 6$
    \item $\mathfrak{g}_{-1}=\langle a_1, a_2\rangle\oplus \langle b_1,\ldots, b_k\rangle$
    \item $l_3(y_{1},y_{2}, y_{3})=a_{1} $,   $l_3(y_{4},y_{5}, y_{6})=a_{2} $ and  $l_1(b_j)=y_{6+j}$ for $j=2,\ldots, k$ and $l_2^m= 0$.
\end{itemize}
Then a Lie 2-action is given by 
\begin{itemize}
    \item $\rho_{1,0}(y_1)=\partial_5$, $\rho_{1,0}(x_1)=e^{-q_5}\partial_1$, $\rho_{1,0}(y_2)=\partial_4$, $\rho_{1,0}(x_2)=e^{-q_4}\partial_2$,  $\rho_{1,0}= 0$ otherwise
    \item  $\rho_2(y_5,y_6)=\partial_1\wedge \partial_2=\rho_{1,-1}(b_k)$ and $\rho_2= 0=\rho_{1,-1}$ otherwise.
\end{itemize}

\noindent The action can be lifted to a fundamental comomentum map defined by: 
\begin{itemize}
    \item ${\lambda}_{1,0}(y_1)=-q_1dq_6+q_3dq_4$, ${\lambda}_{1,0}(y_2)=q_2dq_6-q_3dq_5$ ${\lambda}_{1,0}(x_1)=-e^{-q_5}dq_6$, ${\lambda}_{1,0}(x_2)=-e^{-q_4}dq_6$, ${\lambda}_{1,0}= 0$ otherwise
    \item ${\lambda}_{1,-1}(b_k)=0$, ${\lambda}_{1,-1}=0$ otherwise
    \item ${\lambda}_2(y_1,y_2)=(q_3,0)$ because of $(A_2)$ for ${\lambda}$, ${\lambda}_2= 0$ otherwise. 
\end{itemize}

 \bigskip

 \bigskip

\subsection*{Example 3} 
Lie 2-algebras  $\mathfrak{g}$ satisfying $l_1= 0$ or $l_3=0$ ($\mathfrak{g}_{0}$ is a Lie algebra)

\medskip

\subsubsection*{\underline{Example 3a}:} 

Let $\mathfrak{g}$ satisfy $l_1= 0$ and $l_3\neq 0$.
In this example, which covers the situation $S235$, we construct several Lie 2-actions of the given Lie 2-algebra on $M=\mathbb{R}^3$. We explicit the comomentum map when it exists and give the obstruct of the existence otherwise.

Let us consider the Lie 2-algebra $(\mathfrak{g}= \mathfrak{g}_{-1}\bigoplus\mathfrak{g}_{0}, l_1, l_2, l_3)$ given by $\mathfrak{g_{-1}}=\langle a\rangle_{\mathbb{R}}$ and $\mathfrak{g}_0=\langle x_1,x_2,x_3\rangle_{\mathbb{R}}$ with the following brackets :
\begin{enumerate}
    \item $l_1(a)=0$
     \item $l_2(x_1,x_2)=-x_1$
     \item $l_2(x_1,x_3)=-x_1$
     \item $l_3(x_1,x_2, x_3)=a$.
\end{enumerate}

\medskip

\textbf{Construction of 2-actions}
    
A 2-action of $\mathfrak{g}$ on the manifold $\mathbb{R}^3$ must satisfy: 
 \begin{enumerate}
     \item[$(A_{21})$] $\rho_{1,0}(x_1)=-[\rho_{1,0}(x_1),\rho_{1,0}(x_2)]_S$, 
     \item[$(A_{22})$] $\rho_{1,0}(x_1)=-[\rho_{1,0}(x_1),\rho_{1,0}(x_3)]_S$
      \item[$(A_{23})$] $[\rho_{1,0}(x_2),\rho_{1,0}(x_3)]_S=0$
      \item[$(A_3)$] $[\rho_{1,-1}(a),\rho_{1,0}(x_i)]=0$ for $i=1,2,3$
      \item[$(A_4)$] $\rho_2(x_1,x_2)- \rho_2(x_1,x_3)+\rho_{1,-1}(a)=[\rho_{1,0}(x_1),\rho_2(x_2,x_3)]_S+ [\rho_{1,0}(x_2),\rho_2(x_1,x_3)]_S+$
      
      $[\rho_{1,0}(x_3),\rho_2(x_1,x_2)]_S$.
\end{enumerate}

\begin{rema}
    If $\rho_{1,0}(x_1)=0$ the conditions $(A_{21}), (A_{22})$ and, for $x_1$, $(A_{3})$  are satisfied $\forall \rho_{1,0}(x_2)$, $\forall \rho_{1,0}(x_3)$ and $\forall \rho_{1,-1}(a)$. 
\end{rema}

\medskip

\textbf{Examples of 2-actions}

\begin{multicols}{2}
\begin{enumerate}
    \item[] Sub example (1) with $\rho_{1,0}(x_1)=0$:
\begin{itemize}
    \item $\rho_{1,0}(x_2)=-\rho_{1,0}(x_3)=-\partial_1$
     \item $\rho_2(x_1,x_2)=-\rho_{1,-1}(a)=\partial_1\wedge\partial_3$
      \item $\rho_2(x_1,x_3)=\rho_2(x_2,x_3)=0$
      \item[] 
\end{itemize}
\item[] Sub example (2) with $\rho_{1,0}(x_1)\neq 0$:
\begin{itemize}
    \item $\rho_{1,0}(x_1)=-e^{-q_1}\partial_1$
    \item $\rho_{1,0}(x_2)=-\rho_{1,0}(x_3)=-\partial_1$
     \item $\rho_2(x_1,x_2)=-\rho_{1,-1}(a)=0$
      \item $\rho_2(x_1,x_3)=-e^{-q_1}\partial_1\wedge \partial_2$.
\end{itemize}
\end{enumerate}
\end{multicols}

\begin{rema}
This example covers the situations $T135$, $T146$ and $T235$ since we can also take,

 Sub example (3): $\rho_{1,0}= 0$ and $\rho_2(x_1,x_2)=-\rho_{1,-1}(a)=
 \partial_1\wedge\partial_3$ 

or
 Sub example(4): $\rho_{1,0}(x_1)=0$ with $\rho_{1,0}(x_2)=-\rho_{1,0}(x_3)=-\partial_1$ and $\rho_{1,-1}= 0$ and $\rho_2= 0$. 
\end{rema}

\medskip

\textbf{Action on a 2-plectic manifold}

A Lie 2-action on the 2-plectic manifold
$(\mathbb{R}^3, dq^1\wedge dq^2 \wedge dq^3$) is Hamiltonian if the following closedness conditions are satisfied :
\begin{itemize}
    \item $\iota_{\rho_{1,0}(x_i)}\omega$ are closed $\forall i$
    \item $\iota_{\rho_{1,-1}(a)}\omega$ is closed 
    \item $\iota_{\rho_2(x_i,x_j)}\omega$ are  closed $\forall i,j$.
\end{itemize}

\medskip

\textbf{Explicit 2-comomentum maps}

The above sub example (1) with $\rho_{1,0}(x_1)=0$ satisfies the closedness conditions. The ensuing formulas for ${\lambda}$ are then:
\begin{itemize}
   \item ${\lambda}_{1,0}(x_1)=0$
     \item ${\lambda}_{1,0}(x_2)=-{\lambda}_{1,0}(x_3)=-q_2dq_3$
     \item ${\lambda}_2(x_1,x_2)=-{\lambda}_{1,-1}(a)=(0,(q_2,\partial_1\wedge\partial_3))$
      \item ${\lambda}_2(x_3,x_1)={\lambda}_2(x_3,x_2)=0$.
\end{itemize}

\begin{rema}
    Since for sub example (2), $\rho_{1,0}(x_1)=e^{-q_1}\partial_1$, the form $\iota_{\rho_{1,0}(x_1)}\omega$ is not closed and thus there is no fundamental comomentum map.
\end{rema}

\medskip

\subsubsection*{\underline{Example 3b}:} 

We consider, here a crossed module ($l_1\neq 0, l_3=0$), which refers to Section \ref{strict} and covers the situation $S136$. We provide a Lie 2-action of the given Lie 2-algebra on $M=\mathbb{R}^3$ and give an explicit Lie 2-action with $\rho_2=0$.

Let us consider a Lie algebra $\mathfrak{g}_0=\langle y,x,z\rangle_{\mathbb{R}}$, and a vector space $\mathfrak{g}_{-1}=\langle b,c\rangle_{\mathbb{R}}$ with the relations $[x,y]=x$, $l_1(b)=z$, $l_2^m(b,y)=c$. Then the crossed module $(\mathfrak{g}= \mathfrak{g}_{-1}\oplus\mathfrak{g}_{0}, l_1, l_2, l_3=0)$ is a Lie 2-algebra. 
A Lie 2-action on  $M=\mathbb{R}^3$ is given by:
\begin{itemize}
    \item $\rho_{1,0}(x)=\rho_{1,0}(z)=0$ and no condition on $\rho_{1,0}(y)$
    \item $\rho_{1,-1}(c)=0$ and no condition on $\rho_{1,-1}(b)$
    \item $\rho_2(x,z)=0$
    \item $\rho_2(y,z)=[\rho_{1,-1}(b),\rho_{1,0}(y)]$
    \item there is no condition on $\rho_2(x,y)$.
\end{itemize}   
Thus, we obtain all the situations $T$ (except $T136$). 

Notably, for the situation $T136$, i.e, $\rho_2=0$, we obtain another explicit 2-action given by  $\rho_{1,0}(x)=\rho_{1,0}(z)=0$, $\rho_{1,0}(y)=\partial_1$, $\rho_{1,-1}(c)=\partial_2\wedge\partial_3$ and $\rho_{1,-1}(b)=-q_1\partial_2\wedge\partial_3$.

Observe that the action might not be Hamiltonian. 

\bigskip

\bigskip


\subsection*{Example 4} Lie 2-algebras  $\mathfrak{g}$ satisfying $l_1=0=l_3$ and $l_2\neq 0$

\medskip

We are in the situation $S236$. In example 4a, we construct Lie 2-actions of nilpotent Lie algebras on $\mathbb{R}^3$ and specify the comomentum map when it exists and give the obstruction to the existence otherwise. 
In the second example (Example 4b), we give possible Lie 2-actions with $\rho_{2}=0$.

\medskip

\subsubsection*{\underline{Example 4a}:} 
 
 We first consider the case  $\mathfrak{g}_{-1}=\{0\}$. Thus the Lie 2-algebra reduces to $\mathfrak{g}=\{0\}\oplus\mathfrak{g}_{0}$ where $\mathfrak{g}_{0}$ is a Lie algebra, and $\rho_{1,-1}=0$. This covers the situation $T4$.

 The condition on $\rho: \{0\}\oplus\mathfrak{g}_{0} \longrightarrow \mathfrak{X}^2(M)\oplus\mathfrak{X}^1(M)$ to be a Lie 2-morphism is as follows:
\begin{itemize}
    \item $\rho_{1,0}$ is a Lie morphism 
\item $\rho_2([x,y]_{\mathfrak{g}_0},z)+c.p.=[\rho_{1,0}(x),\rho_2(y,z)]_S +c.p.$
\end{itemize}

\begin{rema}
If $\rho_{1,0}\neq 0$, $\rho_2=\rho_{1,0}\wedge \rho_{1,0}$ does not define a Lie 2-action unless the image of $\rho_{1,0}$ is at most one-dimensional.
    \end{rema}

    \medskip

\textbf{Examples of 2-actions of the Heisenberg Lie algebra }
    
    Let $\mathfrak{g}_{0}=\langle x,y,z\rangle_{\mathbb{R}}$ be the Heisenberg Lie algebra with $[x,y]=z$. A 2-action of $\{0\}\oplus\mathfrak{g}_{0}$ on $(\mathbb{R}^3, dq_1\wedge dq_2\wedge dq_3)$ is given by $\rho_{1,0}(x)=\partial_1$, $\rho_{1,0}(y)=\partial_2$, $\rho_{1,0}(z)=0$ and $\rho_2(y,z)=\rho_2(x,z)=\partial_2\wedge\partial_3$, $\rho_2(x,y)=0$.
    
Then the closedness conditions $d\iota_{\rho_{1,0}(.)}\omega=0$ and  $d\iota_{\rho_{2}(.,.)}\omega=0$  are satisfied and there exists a comomentum map given by 
    \begin{itemize}
        \item ${\lambda}_{1,0}(x)=q_2d_3$, ${\lambda}_{1,0}(y)=q_1d_3$, ${\lambda}_{1,0}(z)=0$
        \item ${\lambda}_2(y,z)={\lambda}_2(x,z)=(0,(q_1,\rho_2(y,z)))$, ${\lambda}_2(x,y)=0$.
    \end{itemize}    This is situation $T145$.

Note that the action given by  $\rho_2(y,z)=q_1\partial_2\wedge\partial_3$, $\rho_2(x,z)=q_2\partial_2\wedge\partial_3$ and $\rho_2(x,y)=0$ can not be lifted to a fundamental comomentum map, since $\iota_{\rho_2(x,z)}\omega=q_2\ dq_1$ is not closed.
   
\begin{rema} We obtain a Lie 2-action of the Lie 2-algebra $\{0\}\oplus\mathfrak{g}_{0}$ if the following three conditions are satisfied: 
\begin{itemize}
\item $\rho_{1,0}$ is a Lie morphism 
    \item $\rho_2([x,y]_{\mathfrak{g}_0},.)= 0,\quad \forall x,y$
    \item $[\mathrm{im}(\rho_{1,0}),\mathrm{im}(\rho_2)]_S=0$.
\end{itemize}  
 \end{rema}
 
    \noindent As an example of this situation, consider the case $\rho_{1,0}= 0$. Then any $\rho_2: \bigwedge^2{\mathfrak{g}_{0}} \longrightarrow \mathfrak{X}^2(M)$ which satisfies $\rho_2([x,y]_{\mathfrak{g}_0},.)=0$ defines a Lie 2-action (compare the two nilpotent examples below). This covers case $T246$.
\medskip

\textbf{Other examples of actions of nilpotent Lie algebras}
        \begin{enumerate}
             \item Let the nilpotent Lie algebra $\mathfrak{g}_0=\langle x_1,x_2,x_3,x_4\rangle_{\mathbb{R}}$ be given by $[x_1,x_2]=x_3$ and $[x_1,x_3]=x_4$, then any $\rho_2: \bigwedge^2{\mathfrak{g}_{0}} \longrightarrow \mathfrak{X}^2(M)$ which satisfies $\rho_2(x_2,x_4)=0$ and $\rho_2(x_3,x_4)=0$ defines a Lie 2-action of $\{0\}\oplus\mathfrak{g}_{0}$ on $M$.
        \item Let the 8-dimensional nilpotent Lie algebra $\mathfrak{g}_0=\langle x_i, i=1\ldots 8\rangle_{\mathbb{R}}$ be given by 
        $[x_1,x_i]=x_{i+1}$, for $2\leq i\leq 7$, $[x_2,x_7]=x_8$, $[x_3,x_6]=-x_7$, and $[x_4,x_5]=x_8$. Then, the condition $\rho_2([x,y]_{\mathfrak{g}_0},.)=0$ reduces to $\rho_2([.,[.,.]_{\mathfrak{g}_0}]_{\mathfrak{g}_0},.)= 0$, since the conditions for $\rho$ to be a Lie 2-action always imply the elements in $[\mathfrak{g}_0,\mathfrak{g}_0]_{\mathfrak{g}_0}$, more precisely those conditions are:
        \begin{multicols}{2}
        \begin{itemize}
            \item[] $ \rho_2(x_2,x_4)=0$, $\rho_2(x_8,x_1)=\rho_2(x_6,x_4)$
             \item[] $\rho_2(x_4,x_6)=\rho_2(x_7,x_1)=\rho_2(x_7,x_3)$
            \item[] $\rho_2(x_3,x_i)=\rho_2(x_{i+1},x_2)$, $4\leq i\leq 6$
            \item[] $\rho_2(x_3,x_7)+\rho_2(x_8,x_1)=\rho_2(x_8,x_2)$.
        \end{itemize}
        \end{multicols}
        \end{enumerate}

        \medskip

\subsubsection*{\underline{Example 4b}:} 

Now we consider the case that $\mathfrak{g}_{-1}$ is possibly non zero but $\rho_2=0$ (situation $T6$).

Let ($\mathfrak{g}_0, [.,.]_{\mathfrak{g}_0}$) be a Lie algebra acting via $\rho_{1,0}$ on a manifold $M$, and consider the Lie 2-algebra $(\bigwedge^2\mathfrak{g}_0\oplus\mathfrak{g}_0, l_1=0, l_2^m(x\wedge y, z)=[x,z]_{\mathfrak{g}_0}\wedge y+ x\wedge [y,z]_{\mathfrak{g}_0}, l_3=0)$. Then $\rho_{1,-1}=\rho_{1,0}\wedge \rho_{1,0}$, together with the given $\rho_{1,0}$ and $\rho_2=0$ defines a Lie 2-action on $M$. 

   \medskip
   
   Let two Lie algebras $\mathfrak{g_0}$ and $\mathfrak{h}$ act on a manifold $M$ by $\nu$ and $\tau$ respectively. Consider the Lie 2-algebra $(\bigwedge^2\mathfrak{h}\oplus\mathfrak{g}_0, l_1=0, l_2^m=0, l_2^p=[.,.]_{\mathfrak{g}_0},l_3=0)$ and define $\rho:\bigwedge^2\mathfrak{h}\oplus\mathfrak{g}_0\longrightarrow \mathfrak{X}^2(M)\oplus\mathfrak{X}^1(M)   $ by $\rho_{1,0}=\nu$, $\rho_2=0$ and $\rho_{1,-1}=\tau\wedge\tau$. Then $\rho$ is a Lie 2-action, if $[\tau(a)\wedge\tau(b),\nu(x)]_S=0$, $\forall a,b\in \mathfrak{h}$ and $\forall x\in \mathfrak{g}_0$.


\subsection*{Example 5} Lie 2-algebra  $\mathfrak{g}$ satisfying $l_1=l_2=l_3=0$\\

This degenerated situation covers $S246$, but we can still construct a comomentum map.

Let us consider the Lie 2-algebra $\mathfrak{g}=\mathfrak{g}_{-1}\oplus\{0\}$ where $n\geq 1$ and $\mathfrak{g}_{-1}=\langle a_1,\ldots, a_n\rangle_{\mathbb{R}}$, with $l_1=l_2=l_3=0$. In this case $\mathfrak{g}_0=\{0\}$ thus $\rho_{1,0}=0$ and thus $\rho_2=0$.\\
Then $\mathfrak{g}$ acts on every manifold  $M$ via $\rho=\rho_{1,-1}:\mathfrak{g}_{-1} \longrightarrow \mathfrak{X}^2(M)$
an arbitrary ${\mathbb{R}}$-linear map.

\noindent When $M$ is 2-plectic, let us write $v_k=\rho_{1,-1}(a_k)$. If for all $k$ there exists a $f_k\in C^\infty_{Ham}(M)$ such that $df_k=-\iota_{v_k}\omega$, we can define a comomentum map ${\lambda}={\lambda}_{1,-1}:\mathfrak{g}_{-1} \longrightarrow C^{\infty}(M)\times Ham^0(M)$ upon setting ${\lambda}(a_k)=(0,(f_k,v_k))$.

This covers situation $T236$.

\vspace{1.5cm}

\subsection*{Acknowledgements.} The authors thank Leonid Ryvkin for helpful discussions and the project \'Emergence Exploratoire OAK (2025) of the Universit\'e de Lorraine for financial support.

\vfill\eject

\bibliographystyle{alpha} 
\bibliography{multisympl-bibdesk}

\end{document}